\documentclass[reqno,11pt]{amsart}

\IfFileExists{mymtpro2.sty}{%
  \usepackage[subscriptcorrection]{mymtpro2}
}{}

\usepackage{a4,amssymb}
\usepackage{graphicx}

\newtheorem{theorem}{Theorem}

\newtheorem{thm}{Theorem}[section]
\newtheorem{lemma}{Lemma}[section]
\newtheorem{follow}{Corollary}[section]
\newtheorem{pr}{Proposition}[section]
\theoremstyle{definition}

\newcommand{\bel}{\begin{equation} \label}
\newcommand{\ee}{\end{equation}}

\theoremstyle{remark}
\newtheorem{remark}[theorem]{Remark}
\newtheorem{myremarks}[theorem]{Remarks}

    {\end{nummer}\end{myremarks}}

\newcounter{numcount}
\newcommand{\labelnummer}{\mbox{\normalfont (\roman{numcount})}}%

\makeatletter

\newenvironment{nummer}%
  {\let\curlabelspeicher\@currentlabel%
    \begin{list}{\labelnummer}%
      {\usecounter{numcount}\leftmargin0pt%
        \topsep0.5ex\partopsep2ex\parsep0pt\itemsep0ex\@plus1\p@%
        \labelwidth2.5em\itemindent3.5em\labelsep1em%
      }%
    \let\saveitem\item%
    \def\item{\saveitem%
      \def\@currentlabel{{\upshape\curlabelspeicher}$\,$\labelnummer}}%
    \let\savelabel\label%
    \def\label##1{\savelabel{##1}%
      \@bsphack%
        \ifmmode\else%
          \protected@write\@auxout{}%
          {\string\newlabel{##1item}{{\labelnummer}{\thepage}}}%
        \fi%
      \@esphack%
    }%
  }{\end{list}}%

\renewcommand{\appendix}{\def\thesection{\textsc{Appendix}}}

 \let\leq\le
 \let\geq\ge

\DeclareMathOperator{\tr}{tr\kern1pt}

\newcommand{\Dom}{\mathop\mathrm{Dom}\nolimits}

\newcommand{\Ran}{\mathop\mathrm{Ran}\nolimits}

\newcommand{\supp}{\mathop\mathrm{supp}\nolimits}

\newcommand{\eps}{\varepsilon}
\newcommand{\EE}{\mathcal{E}}
\newcommand{\FF}{\mathcal{F}}
\newcommand{\HH}{\mathcal{H}}
\newcommand{\UU}{\mathcal{U}}
\newcommand{\VV}{\mathcal{V}}
\newcommand{\ef}{\mathfrak{e}}
\newcommand{\re}{{\mathbb R}}
\newcommand{\N}{{\mathbb N}}

\newcommand{\proj}{{\mathbb P}}

\makeatletter

\newif\ifper\pertrue
\def\per{.}

\def\bti{\@ifnextchar[\bbti\bbbti}
\def\bbti[#1]#2{#2, #1.}
\def\bbbti#1{#1.}

\def\z{\@ifnextchar[\zz\zzz}
\def\zz[#1]#2#3#4#5{\perfalse\emph{#2} \textbf{#3}, #4 (#5) [#1]}
\def\zzz#1#2#3#4{\emph{#1} \textbf{#2}, #3 (#4)\ifper\per\fi\pertrue}

\def\pub{\@ifstar\pubstar\pubnostar}
\def\pubnostar{\@ifnextchar[\@@pubnostar\@pubnostar}
\def\@@pubnostar[#1]#2#3#4{#2, #3, #4, #1\ifper\per\fi\pertrue}
\def\@pubnostar#1#2#3{#1, #2, #3\ifper\per\fi\pertrue}
\def\pubstar[#1]#2#3#4{\perfalse #2, #3, #4 [#1]\pertrue}

\makeatother

\newcommand{\beq}{\begin{equation}}
\newcommand{\eeq}{\end{equation}}
\newcommand{\ba}{\begin{array}}
\newcommand{\ea}{\end{array}}
\newcommand{\bea}{\begin{eqnarray}}
\newcommand{\eea}{\end{eqnarray}}
\newcommand{\beas}{\begin{eqnarray*}}
\newcommand{\eeas}{\end{eqnarray*}}
\newcommand{\Pre}[1]{\ensuremath{\mathrm{Re} \left( #1 \right)}}

{

\newcommand{\R}{\mathbb{R}}

\def\P{I\kern-.30em{P}}
\def\E{I\kern-.30em{E}}
\renewcommand{\E}{\mathbb{E}\mkern2mu}
\renewcommand{\P}{\mathbb{P}}

\begin{document}

\title[Edge currents for magnetic barriers]{Edge currents and eigenvalue estimates for magnetic barrier Schr\"odinger operators}

\author[N.\ Dombrowski]{Nicolas Dombrowski}
\address{Department of Mathematics,
    University of Helsinki,
    FI-00014 Helsinki, Finland }
\email{nicolas.dombrowski@helsinki.fi}

\author[P.\ D.\ Hislop]{Peter D.\ Hislop}
\address{Department of Mathematics,
    University of Kentucky,
    Lexington, Kentucky  40506-0027, USA}
\email{hislop@ms.uky.edu}

\author[E.\ Soccorsi]{Eric Soccorsi}
\address{CPT, CNRS UMR 7332, Aix Marseille Universit\'e, 13288 Marseille, France \&
Universit\'e du Sud Toulon-Var, 83957 La Garde, France}
\email{eric.soccorsi@univ-amu.fr}


\begin{abstract}
   We study two-dimensional magnetic Schr\"odinger operators with a magnetic field that is equal to $b>0$ for $x > 0$ and $-b$ for $x < 0$.
  This magnetic Schr\"odinger operator exhibits a magnetic barrier at $x=0$. The unperturbed system is invariant with respect to translations in the $y$-direction. As a result, the Schr\"odinger operator admits a
  direct integral decomposition. We analyze the band functions of the fiber operators as functions of the wave number and establish their asymptotic behavior. Because the fiber operators are reflection symmetric, the band functions may be classified as odd or even. The odd band functions have a unique absolute minimum. We calculate the effective mass at the minimum and prove that it is positive.
  The even band functions are monotone decreasing. We prove that the eigenvalues of an Airy operator, respectively, harmonic oscillator operator, describe the asymptotic behavior of the band functions for large negative, respectively positive, wave numbers.
 We prove a Mourre estimate for a family of magnetic and electric perturbations of the magnetic Schr\"odinger operator and establish the existence of absolutely continuous spectrum in certain energy intervals. We prove lower bounds on magnetic edge currents for states with energies in the same intervals. For a different class of perturbations, we also prove that these lower bounds imply stable lower bounds for the asymptotic edge currents. We study the perturbation by slowly decaying negative potentials. Using the positivity of the effective mass, we establish the asymptotic behavior of the eigenvalue counting function for the infinitely-many eigenvalues below the bottom of the essential spectrum.
\end{abstract}

\maketitle \thispagestyle{empty}


\tableofcontents

\vspace{.2in}

{\bf  AMS 2000 Mathematics Subject Classification:} 35J10, 81Q10,
35P20\\
{\bf  Keywords:}
magnetic Schr\"odinger operators, snake orbits, magnetic field, magnetic edge states, edge conductance\\


\section{Statement of the problem and results}\label{sec:introduction}
\setcounter{equation}{0}

We continue our analysis of the spectral and transport properties of perturbed magnetic Schr\"odinger
operators describing electrons in the plane moving under the influence of a transverse magnetic field.
In \cite{his-soc1}, two of us studied the original Iwatsuka model for which $0 < b_- < b_+ < \infty$.
The basic model treated in this paper consists of a transverse magnetic field that is constant in each half plane so that it is
equal to $b > 0$ for $x > 0$ and $- b < 0$ for $x<0$. We choose a gauge so that the corresponding vector potential has the form $(0, A_2(x,y))$.
The second component of the vector potential $A_2 (x,y)$ is obtained
by integrating the magnetic field so that $A_2(x,y) = b |x|$, independent of $y$.
The fundamental magnetic Schr\"odinger operator is:
\beq\label{eq:basic1}
H_0 := p_x^2 + ( p_y - b|x|)^2 , ~~p_x := -i \partial/ \partial x, ~~p_y := - i \partial / \partial y ,
\eeq
defined on the dense domain $C_0^\infty ( \R^2) \subset {\rm L}^2 (\R^2)$. This operator extends to a nonnegative self-adjoint operator in ${\rm L}^2 (\R^2)$.

The magnetic field is piecewise constant and equals $\pm b$ on the half-planes $\R_\pm^* \times \R$, where $\R_\pm^* := \R_\pm \backslash \{ 0 \}$.
The discontinuity in the magnetic field at $x=0$ is called a \emph{magnetic edge}. Classically, a particle moving within a distance of $\mathcal{O}(b^{-1/2})$ of the edge moves in a \emph{snake orbit} \cite{reijniers-peeters}.
Half of a snake orbit lies in the half-plane $x>0$,
 and the other half of the orbit lies in $x < 0$. We prove that the quantum model has current flowing along the magnetic edge at $x=0$
 and that the current is localized in a small neighborhood of size $\mathcal{O}(b^{-1/2})$ of $x=0$.


\subsection{Fiber operators and reflection symmetry}\label{subsec:fibre1}

Due to the translational invariance in the $y$-direction, the operator $H_0$ on ${\rm L}^2 (\R^2)$ is unitarily equivalent to the direct integral of
operators $h(k)$, $k \in \R$, acting on ${\rm L}^2 (\R)$.
This reduction is obtained using the partial Fourier transform with respect to the $y$-coordinate and defined as
$$
(\FF u)(x,k) = \hat{u}(x,k): = \frac{1}{\sqrt{2\pi}} \int_{\re} \mathrm{e}^{-i y k}
u(x,y) dy,~~\ (x,k) \in \re^2.
$$
Then we have $\FF H_0 \FF^* = {\mathcal H}_0$ where
$$
{\mathcal H}_0 : =  \int_{\re}^\oplus h(k) dk,
$$
and the fiber operator $h(k)$ acting in $\HH:=\mathrm{L}^2(\re)$ is
$$
h(k)  := p_x^2 + (k - b|x|)^2,~~\ k \in \re .
$$

Since the effective potential $(k - b|x|)^2$ is unbounded as $|x| \rightarrow \infty$,
the self-adjoint fiber operators $h(k)$ have compact resolvent. Consequently, the spectrum of $h(k)$ is discrete.
We write $\omega_j (k)$ for the eigenvalues listed in increasing order. They are all simple (see \cite[Appendix: Proposition A.2]{HS1})
and depend analytically on $k$.
As functions of $k \in \R$, these functions are called the {\it band functions} or {\it dispersion curves} and their properties
play an important role. For fixed $k \in \R$, we denote by $\psi_j(k)$ the ${\rm L}^2$-normalized eigenfunctions of $h(k)$ with eigenvalue $\omega_j(k)$.
These satisfy the eigenvalue equation:
\beq\label{eq:ev-eqn1}
h(k) \psi_j (x,k) = \omega_j (k) \psi_j(x,k), ~~\psi_j(x,k) \in {\rm L}^2 (\R), ~~\| \psi_j (\cdot,k) \| = 1 .
\eeq
We choose all $\psi_j(k)$ to be real, and $\psi_1(x,k) > 0$ for $x \in \re$ and $k \in \re$. The rank-one orthogonal projections $P_j(k) := \langle \cdot, \psi_j(k) \rangle \psi_j(k)$, $j \in \N^*$, depend analytically on $k$ by standard arguments.

The full operator $H_0$ exhibits reflection symmetry with respect to $x=0$.
Let $I_P$ be the parity operator:
\beq\label{eq:parity1}
(I_Pf)(x,y) := f(-x,y),
\eeq
so that $I_P^2 = 1$. The Hilbert space ${\rm L}^2 (\R^2)$ has an orthogonal decomposition
corresponding to the eigenspaces of $I_P$ with eigenvalue $\pm 1$. The Hamiltonian $H_0$ commutes with $I_P$ so each eigenspace of $I_P$
is an $H_0$-invariant subspace.

This symmetry passes to the fiber decomposition. For each $k \in \R$ we have $[ h(k) , I_P] = 0$, where $I_P$ is the restriction to ${\rm L}^2 (\R)$ of the operator defined in \eqref{eq:parity1}. Since the eigenvalues of $h(k)$ are simple, for each $k \in \R$, there is a map $\theta_j(k): \re \to \{ \pm 1 \}$
so that
$$
(I_P \psi_j)(x,k) = \theta_j(k) \psi_j(x,k),~~\ k \in \re,~~\ j \in  \N^*,
$$
as $\psi_j(x,k)$ is ${\rm L}^2(\R_x)$-normalized and real-valued. We show that $\theta_j(k)$ is independent of $k$.
Since the mapping $k \mapsto P_j(k)$, the orthogonal projector onto $\psi_j(\cdot,k)$,
is analytic, it follows that $\theta_j(k)= \theta_j(0)$ for every $k \in \re$.
Consequently, each eigenfunction $\psi_j(x , k)$ is either even or odd in $x$.

We have an $h(k)$-invariant
decomposition ${\rm L}^2 (\R) = \mathcal{H}_- \oplus \mathcal{H}_+$, according to the eigenvalues $\{ -1, +1\}$ of the
projection $(I_P f)(x) = f(-x)$.
From this then follows that
$h(k)=h^+(k) \oplus h^-(k)$, where
$$
h^{\pm}(k):=h(k)_{| \HH_{\pm}},\ \HH_{\pm}:= \{ f \in \HH,\ I_P f = \pm f \}.
$$
We analyze the spectrum of $h(k)$ by studying the spectrum of the restricted operators letting
$\sigma(h^{\pm}(k)) := \{ \omega_j^{\pm}(k),\ j \in \N^* \}$.
Bearing in mind that $\omega_j^+(k) < \omega_j^-(k)$ and $\sigma(h(k)) = \sigma(h^+(k)) \cup \sigma(h^-(k))$ for every $k \in \re$, we have
$$
\omega_j^+(k)=\omega_{2j-1}(k),\ \omega_j^-(k)=\omega_{2j}(k),\ j \in \N^*.
$$


\subsection{Effective potential}\label{subsec:eff-pot1}

The fiber operator $h(k)$ has an effective potential:
$$
V_{eff} (x,k) := ( k - b |x|)^2, ~~x, k  \in \R .
$$
The properties of this potential determine those of the band functions.

\vspace{.1in}
\noindent
{\em Positive $k > 0$}. There are two minima of $V_{eff}$ at $x_\pm := \pm k / b$. The potential consists of two parabolic potential
wells centered at $x_\pm$ and has value $V_{eff}(0,k) = k^2$. As $k \rightarrow + \infty$, the potential wells separate and the barrier between the two minima grows to infinity.
\vspace{.1in}

\noindent
{\em Negative $k< 0$}. The effective potential is a parabola centered at $x=0$ and $V_{eff}(0,k) = k^2$ is the minimum. Consequently, as $k \rightarrow - \infty$, the minimum of this potential well goes to plus infinity.


\subsection{Band functions}\label{subsec:disp-curves1}

The behavior of the effective potential determines the band functions. For $k > 0$, the symmetric double wells of $V_{eff}$
indicate that there are two eigenvalues near each level of a harmonic oscillator Hamiltonian. The splitting of these eigenvalues is
exponentially small in the tunneling distance in the Agmon metric between $x_\pm$. As $k \rightarrow + \infty$, this tunneling effect
is suppressed and these two eigenvalues approach the harmonic oscillator eigenvalue exponentially fast. For $k < 0$, there is a single potential well with a minimum that goes to infinity as $k \rightarrow - \infty$. Hence, the band functions
diverge to plus infinity in this limit. Several band functions along with the parabola $E = k^2$ are shown in Figure 1.

\subsection{Relation to edge conductance}\label{subsec:relation1}

Dombrowski, Germinet, and Raikov \cite{DGR} studied the quantization of the Hall edge conductance for a generalized family of Iwatsuka models including the model discussed here. Let us recall that the Hall edge conductance is defined as follows.
We consider the situation where the edge lies along the $y$-axis as discussed above.
Let $I := [a,b] \subset \R$ be a compact energy interval. We choose a smooth decreasing function $g$
so that ${\rm supp} ~g' \subset [a, b ]$.
Let $\chi = \chi (y)$ be an $x$-translation invariant smooth function with $\mbox{supp} ~\chi' \subset [-1/2, 1/2]$. The edge Hall conductance is
defined by
$$
\sigma_e^I (H) := - 2 \pi {\rm tr} ~ ( g'(H) i[ H, \chi ] ) ,
$$
whenever it exists. The edge conductance measures the current across the axis $y=0$ with energies below the energy interval $I$.

Theorem 2.2 of \cite{DGR} presents the quantization of edge currents for the generalized Iwastuka model.
For this model, the magnetic field $b(x)$ is simply assumed to be monotone and to have values $b_\pm$ at $\pm \infty$.
The energy interval $I$ is assumed to satisfy the following condition. There are two nonnegative integers $n_\pm \geq 0$ for which
\beq\label{eq:gen-iwatsuka1}
I \subset ( (2n_- - 1)|b_-|, (2n_- + 1)|b_-|) \cap ( (2n_+ - 1)|b_+|, (2 n_+ + 1)|b_+|), ~~n_\pm \neq 0 .
\eeq
If $n_\pm = 0$, the corresponding interval should be taken to be $(- \infty, |b_\pm|)$.
Under condition \eqref{eq:gen-iwatsuka1}, Dombrowski, Germinet, and Raikov \cite{DGR} proved
$$
\sigma_e^I (H) = ({\rm sign}~ b_-) n_- - ({\rm sign}~ b_+) n_+ .
$$
Applied to the model studied here where $b_+ > 0$ and $b_- = - b_+ < 0$, and under condition \eqref{eq:gen-iwatsuka1}, we have
$$
\sigma_e^I (H) = - ( n_-  + n_+) .
$$
In particular, if $b_+ = b > 0$, and $I \subset ( (2n - 1)b, (2n + 1)b )$, we have $\sigma_e^I (H) = - 2n$.


We complement this result by proving in sections \ref{sec:mourre-est1} and \ref{sec:edge-curr1} the existence
and localization of edge currents for $H_0$ and its perturbations. Following the notation of those sections, we prove,
roughly speaking, that there is a nonempty interval $\Delta_E (\delta_0)$
between the Landau levels $(2n-1)b$ and $(2n+1)b$
   and a finite constant $c_n > 0$, so that for any state $\psi = \mathbb{P}_0 (\Delta_E(\delta_0)) \psi$, where $\mathbb{P}_0(\Delta_E(\delta_0))$ is the spectral projector for $H_0$ and the interval $\Delta_E(\delta_0)$,
we have
$$
\langle \psi, v_y \psi \rangle \geq \frac{c_n}{2} b^{1/2} \| \psi \|^2 > 0, ~~ v_y := - (p_y - b|x|) .
$$
This lower bound indicates that such a state $\psi$ carries a nontrivial edge current for $H_0$.
We prove that this estimate is stable for a family of magnetic and electric perturbations of $H_0$.


\subsection{Contents}\label{subsec:contents1}
We present the properties of the band functions $\omega_j(k)$ for the unperturbed fiber operator $h(k)$ in section \ref{sec:dispersion-curves1}. The emphasis is on the behavior of the band functions as $k \rightarrow \pm \infty$. The basic Mourre estimate for the unperturbed operator $H_0$ is derived in section \ref{sec:mourre-est1} and its stability under perturbations is proven. As a consequence, this shows that there is absolutely continuous spectrum in certain energy intervals. Existence, localization, and stability of edge currents for a family of electric and magnetic perturbations is established in section \ref{sec:edge-curr1}. These edge currents and their lower bounds are valid for all times. We also prove a lower bound on the asymptotic velocity for a different class of perturbations in Theorem \ref{th:asymptotic1}. In section \ref{sec:ev-counting1}, we study perturbations by negative potentials decaying at infinity. We demonstrate that such potentials create infinitely-many eigenvalues that accumulate at the bottom of the essential spectrum from below. We establish the asymptotic behavior of the eigenvalue counting function for these eigenvalues accumulating at the bottom of the essential spectrum.


\subsection{Notation}\label{subsec:note1}

We write $\langle \cdot, \cdot \rangle$ and $\| \cdot \|$ for the inner product and norm on ${\rm L}^2(\R^2)$.
The functions are written with coordinates $(x,y)$, or, after a partial Fourier transform with respect to $y$,
we work with functions $f(x,k) \in {\rm L}^2(\R^2)$. We often view these functions $f(x,k)$ on ${\rm L}^2 (\R_x)$ as parameterized by $k \in \R$.
In this case, we also write $\langle f(\cdot, k), g(\cdot, k ) \rangle$ and $\| f(\cdot, k) \|$ for the inner product
and related norm on ${\rm L}^2 (\R_x)$. So whenever an explicit dependance on the parameter $k$ appears, the functions should be considered on ${\rm L}^2 (\R_x)$. We indicate explicitly in the notation, such as $\| \cdot \|_X$, for $X = {\rm L}^2 (\R_\pm)$, when we work on those spaces. We write $\| \cdot \|_\infty$ for $\| \cdot \|_{{\rm L}^\infty (X)}$ for $X= \R, \R_\pm, ~{\rm or} ~ \R^2$.
For a subset $X \subset \R$, we denote by $X^*$ the set $X^* := X \backslash \{ 0 \}$. Finally for all $n \in \N$ we put $\N_n := \{ j \in \N,\ j \leq n \} = \{ 0,1,\ldots, n \}$.

\subsection{Acknowledgements}\label{subsec:acknowledgements1}

ND is supported by the Center of Excellence in Analysis and Dynamics Research of the Finnish Academy.
PDH thanks the Universit\'e de Cergy-Pontoise and the Centre de Physique Th\'eorique, CNRS, Luminy, Marseille, France, for its hospitality. PDH was partially supported by the Universit\'e du Sud Toulon-Var, La Garde, France, and National Science Foundation grant 11-03104
during the time part of the work was done. ES thanks the University of Kentucky for hospitality.

\begin{remark}
After completion of this work, we learned of a similar analysis of the band structure by Nicolas Popoff \cite{popoff} in his 2012 thesis at the Universit\'e Rennes I. We thank Nicolas for many discussions and for letting us use his graph in Figure 1.
\end{remark}

\begin{remark} After completing this paper, we discovered the paper ``Dirichlet and Neumann eigenvalues for half-plane magnetic
Hamiltonians,'' by V.\ Bruneau, P.\ Miranda, and G.\ Raikov \cite{BMR}. Their Corollary 2.4, part (i), is similar to our Theorem \ref{thm-ea}.
\end{remark}

\section{Properties of the band functions}\label{sec:dispersion-curves1}
\setcounter{equation}{0}

In this section, we prove the basic properties of the band functions $k \in \R \mapsto \omega_j (k)$.
We have the basic identity:
$$
\omega_j (k) = \langle \psi_j (\cdot, k) , h(k) \psi_j ( \cdot, k) \rangle .
$$

According to section \ref{subsec:fibre1}, the eigenfunctions of $h(k)$ are either
\textbf{even} and lie in $\mathcal{H}_+$, or \textbf{odd} and lie in $\mathcal{H}_-$, with respect to the reflection $x \mapsto - x$.
We label the states so that the eigenfunctions $\psi_{2j-1} \in \mathcal{H}_+$ and $\psi_{2j} \in \mathcal{H}_-$, for $j = 1,2,3, \ldots$.
The restrictions of $h(k)$ to $\mathcal{H}_\pm$ are denoted by $h_\pm (k)$, with eigenvalues $\omega_j^+ (k) =
\omega_{2j-1}(k)$ and $\omega_j^- (k) = \omega_{2j}(k)$, respectively.

Following the qualitative description in section \ref{subsec:eff-pot1}, we have the following asymptotics for the band functions.
When $k \rightarrow  + \infty$, the band function satisfies $\omega_j(k) \rightarrow (2j-1)b$, whereas as $k \rightarrow - \infty$,
we have $\omega_j (k) \rightarrow + \infty$.

\begin{pr}\label{prop:disp-deriv1}
The band functions $\omega_j(k)$ are differentiable and the derivative satisfies
\beq\label{eq:ev-deriv1}
\omega_j ' (k) = \frac{-2}{b} \left[ ( \omega_j (k) - k^2) \psi_j (0,k)^2 + \psi_j ' (0,k)^2 \right] .
\eeq
As a consequence, we have a classification of states:
\begin{enumerate}
\item {\bf Odd states:} $\psi_{2j} (0,k) = 0$. The band functions satisfy:
\beq\label{eq:disp-odd1}
(\omega_j^-)' (k) = \omega_{2j}'(k) = \frac{-2}{b} \psi_{2j} ' (0,k)^2 < 0 .
\eeq
\item {\bf Even states:} $\psi_{2j-1} ' (0,k) = 0$. The band functions satisfy:
\beq\label{eq:disp-even1}
(\omega_j^+)' (k) = \omega_{2j-1} ' (k) = \frac{-2}{b} ( \omega_{2j-1} (k) - k^2) \psi_{2j-1} (0,k)^2.
\eeq
\end{enumerate}
\end{pr}

\begin{proof}
The Feynman-Hellmann Theorem gives us
\beas
\omega_j ' (k) &=& \int_\R ~2 (k - b|x|) \psi_j(x,k)^2 ~dx \nonumber \\
 &=& \frac{-1}{b} \int_0^\infty \psi_j (x,k)^2 \frac{d}{dx} (k-bx)^2 ~dx \nonumber \\
  & & +  \frac{1}{b} \int_{-\infty}^0 \psi_j (x,k)^2 \frac{d}{dx} (k+bx)^2 ~dx .
\eeas
Integrating by parts, and using the ordinary differential equation \eqref{eq:ev-eqn1}, we obtain \eqref{eq:ev-deriv1}.
Note that $\lim_{x \rightarrow \pm \infty} x^2 \psi_j(x,k)^2 = 0$ since $\psi_j(k)$ is in the domain of $h(k)$ (see \cite[Lemma 3.5]{iwatsuka1}).
\end{proof}



\begin{figure}
  \centering
\begin{center}
\includegraphics[height=80mm,angle=0]{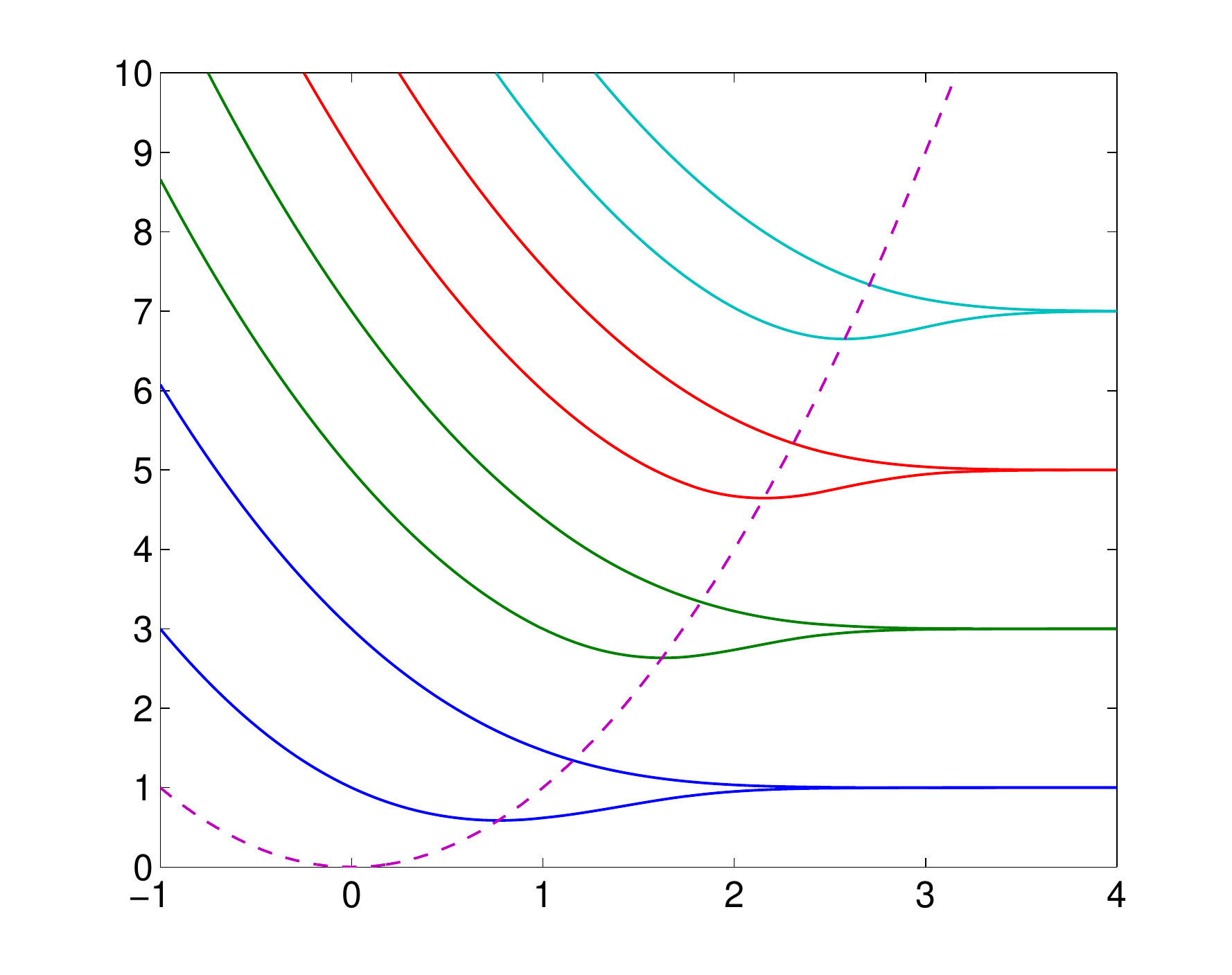}
\end{center}
  \caption{Approximate shape of the band functions $k \mapsto \omega_j(k)$, for $j=1, \ldots, 8$, of the Iwatsuka Hamiltonian with $-b < 0 < b$ and $b=1$. The dotted curve is $E = k^2$. Graph courtesy of N.\ Popoff.}
\end{figure}



Let us note that we cannot have both $\psi_j(0,k)=0$ and $\psi_j' (0,k)=0$.
As consequences, the band functions for odd states are strictly monotone decreasing $\omega_{2j}^\prime (k) < 0$.
For even states, there is a minimum at $k = \kappa_j$ satisfying
$$
\omega_{2j-1} (\kappa_j) = \kappa_j^2 .
$$
We will prove in Proposition \ref{prop:effective-mass1} that this
is the unique critical point of these band functions and that it is a non-degenerate minimum.
This shows that there is an effective mass at this point. This is essential for the discussion in section \ref{sec:ev-counting1}.


\subsection{Absolutely continuous spectrum for $H_0$}\label{subsec:ac-spectrum1}

The spectrum of $H_0$ is the union of the ranges of the band functions $\sigma (H_0) = \cup_{j \geq 1} \omega_j (\R) = [ \omega_1(\kappa_1), \infty )$.
The band functions are analytic and nonconstant by Propositions \ref{prop:disp-deriv1} and \ref{prop:effective-mass1}. Consequently, from \cite[Theorem XIII.86]{RS-IV},
the spectrum of $H_0$ is purely absolutely continuous.


\subsection{Band function asymptotics $k \rightarrow - \infty$.}

As $k \rightarrow - \infty$, we will prove that the fiber Hamiltonian $h(k)$ is well approximated by an Airy operator
\beq\label{eq:airy1}
h_{\rm Ai}(k) := p_x^2 + 2 b |k| |x| + k^2 ,
\eeq
in the sense that the band functions of $h(k)$ are close to the band functions of the Airy operator $h_{\rm Ai}(k)$.
In order to establish this, let $Ai(x)$ be the standard Airy function whose zeros are located on the negative real axis. The Airy function satisfies the Airy ordinary differential equation:
$$
Ai''(x) = x Ai(x).
$$
By scaling and translations, it follows that the Airy function $Ai(\gamma x + \sigma)$ satisfies
\beq\label{eq:airy3}
(p_x^2 + \gamma^3 x) Ai(\gamma x + \sigma) = - \gamma^2 \sigma Ai(\gamma x + \sigma), ~~\gamma, \sigma \in \R.
\eeq

The model Airy Hamiltonian $h_{\rm Ai}(k)$ in \eqref{eq:airy1}
has discrete spectrum $\tilde{\omega}_j(k)$ and eigenfunctions $\tilde{\Psi}_j^{\rm Ai}(x,k)$ satisfying
\beq\label{eq:airy4}
h_{\rm Ai}(k) {\tilde \Psi}_j^{\rm Ai}(x,k) = \tilde{\omega}_j(k) {\tilde \Psi}_j^{\rm Ai} (x,k) .
\eeq
It follows from \eqref{eq:airy3} that the eigenfunction ${\tilde \Psi}_j^{\rm Ai}(x,k)$
in \eqref{eq:airy4} is a multiple of the scaled and translated Airy function. The non-normalized
solution ${\tilde \Psi}_j^{\rm Ai} (x,k)$ for the eigenvalue $\tilde{\omega}_j(k)$ is
$$
{\tilde \Psi}_j^{\rm Ai}(x,k) = Ai \left( (2 b |k|)^{1/3} |x| + \frac{k^2 -
{\tilde \omega}_j(k)}{(2 b |k|)^{2/3}} \right),\ k < 0,\ x \in \R,
$$
with an eigenvalue given by
$$
\tilde{\omega}_j(k)  =  k^2 - (2 b |k|)^{2/3} \sigma .
$$


We determine $\sigma$ as follows. The operator $h_{\rm Ai}(k)$ commutes with the parity operator $I_P$ so its states are even or odd.
The \emph{odd eigenfunctions} $\Psi_j^{Ai,o}(x,k)$
of $h_{\rm Ai}(k)$ must satisfy $\Psi_j^{\rm Ai, o}(0,k) =0$. Consequently, the ${\rm L}^2 (\R)$-normalized
odd eigenfunctions $\Psi_j^{\rm Ai, o} (x,k) = \tilde{\Psi}_{2j}^{Ai}(x,k)$ are given by
\beq\label{eq:airy-kneg1}
\Psi_j^{\rm Ai, o} (x,k) = C_{Ai,j}(b,k) ( {\rm sign} ~x ) Ai ( (2 b |k|)^{1/3} |x| +  z_{Ai, j} ) , ~~\Psi_j^{\rm Ai,o} (0,k) = 0,
\eeq
where $z_{Ai,j}$ is the $j^{\rm th}$ zero of $Ai(x)$
and the corresponding eigenvalue is
$$
\tilde{\omega}_{2j}(k) = k^2 - (2 b |k|)^{2/3} z_{Ai,j} .
$$
The \emph{even eigenfunctions} $\Psi_j^{\rm Ai, e} (x,k) = \tilde{\Psi}_{2j-1}^{Ai}(x,k)$ of $h_{\rm Ai}(k)$ must have a vanishing derivative at $x=0$ and are given by
\beq\label{eq:airy-kneg2}
\Psi_j^{\rm Ai, e} (x,k) = C_{Ai',j}(b,k) Ai ( (2 b|k|)^{1/3} |x| +  z_{Ai', j} ) , ~~( \Psi_j^{\rm Ai, e} )^\prime  (0,k) = 0,
\eeq
and the corresponding eigenvalue is
$$
\tilde{\omega}_{2j-1}(k) = k^2 - (2 b |k|)^{2/3} z_{Ai',j} ,
$$
where $z_{Ai',j}$ is the $j^{\rm th}$ zero of $Ai^\prime (x)$.
The normalization constant $C_{X,j}(b,k)$, for $X = Ai ~{\rm or} ~Ai'$ is given by
\beq\label{eq:airy-normal1}
C_{X,j}(b,k) := \left( \frac{(2 b |k|)^{1/3}}{2 c_{X,j}} \right)^{1/2}, {\rm where} ~~c_{X,j} := \int_0^\infty Ai(v + z_{X,j})^2 ~dv.
\eeq

We now obtain estimates on the band functions $\omega_j(k)$ as $k \rightarrow - \infty$.

\begin{pr}\label{prop:dispersion-kminusinfty1}
For each $j \in \N^*$, as $k \rightarrow - \infty$, we have
\beq\label{eq:airy-est1}
\| ( h(k) - [k^2 - (2 b |k|)^{2/3} z_{X, j}]) \Psi_j^{\rm Ai,u} (\cdot,k) \| \leq \frac{b^{4/3}}{(2 |k|)^{2/3}} D_{X,j},
\eeq
where the constant $D_{X,j}$, given in \eqref{eq:airy-est5},
is independent of the parameters $(k,b)$, and $(X , u) = (Ai , e)$ or $(Ai',o)$, for even or odd states, respectively.
This immediately implies the eigenvalue estimate
$$
| \omega_j(k) - [k^2 - (2 b |k|)^{2/3} z_{X, j}] |  \leq \frac{b^{4/3}}{(2 |k|)^{2/3}} D_{X,j},
 ~~k \rightarrow - \infty.
$$
\end{pr}

\begin{proof}
In order to prove \eqref{eq:airy-est1}, we note that
$$
h(k) - h_{\rm Ai}(k) = b^2 x^2,
$$
so that with the definition of $\Psi_j^{\rm Ai,u}(x,k)$ in \eqref{eq:airy-kneg1} for $u=o$ and \eqref{eq:airy-kneg2} for $u=e$, and the normalization constant $C_{X,j}$ in \eqref{eq:airy-normal1},
we have
\beas
\| [ h(k) - h_{\rm Ai}(k) ] \Psi_j^{\rm Ai,u}(\cdot,k)\|^2 & =&  \frac{ 2  b^4}{(2 |k| b)^{5/3}} C_{X,j}^2
 \int_0^\infty v^4 Ai(v+z_{X,j})^2 ~dv \nonumber \\
 &=& \frac{b^{8/3}}{(2 |k|)^{4/3}} D_{X,j}^2,
\eeas
where the constant $D_{X,j}$, given by
\beq\label{eq:airy-est5}
D_{X,j} := \left( \frac{\int_0^\infty v^4 Ai(v+z_{X,j})^2 ~dv}{c_{X,j}} \right)^{1 \slash 2},
\eeq
is finite since $Ai(v) \sim e^{-v^{3/2}}$ as $v \rightarrow + \infty$.
\end{proof}


\subsection{Band functions asymptotics $k \rightarrow + \infty$}\label{subsec:positive-k1}

For $k \geq 0$, the effective potential consists of two double wells that separate as $k \rightarrow +\infty$. Consequently $\omega_j^+ (k)$ approaches $\omega_j^- (k)$ as $k \rightarrow + \infty$. The eigenvalues of the double well potential consists of pairs of eigenvalues whose differences are exponentially small as $k \rightarrow \infty$. Thus, the effective Hamiltonian for $k = + \infty$ is the harmonic oscillator Hamiltonian:
$$
h_{\rm HO}(k):=-\frac{d^2}{dx^2} + (bx-k)^2 .
$$
We let $\ef_0(b):=0$ and $\ef_j(b):=(2j-1)b$, for every $j \in \N^*$, denote the energy levels of the harmonic oscillator.
Let $\Psi_j^{\rm HO} (k)$ denote the $j^{\rm \tiny th}$ normalized eigenfunction of the harmonic oscillator
so that $h_{\rm HO}(k)\Psi_j^{\rm HO}(k)=\ef_j(b) \Psi_j^{\rm HO}(k)$. It can be explicitly expressed as
\bel{a2}
\Psi_j^{\rm HO} (x,k):=\frac{1}{(2^j j!)^{1 \slash 2}} \left( \frac{b}{\pi} \right)^{1 \slash 4} {\rm e}^{-b \slash 2(x-k \slash b)^2} H_{j}(b^{1 \slash 2}(x-k \slash b)),
\ee
where $H_j$ is the $j^{\rm \tiny th}$ Hermite polynomial.

\begin{pr}\label{prop:kplusinfty1}
For each $j \in \N$, there exists a constant  $0 < C_j < \infty$, depending only on $j$, so that for $k \geq 0$, we have,
\beq\label{eq:ho4.1}
\| (h(k)-\ef_j(b)) \Psi_j^{\rm HO}(\pm x;k) \| \leq C_j b {\rm e}^{-k^2 \slash (4b)} .
\eeq
This immediately implies the eigenvalue estimate
\bel{a0}
0 < \mp (\omega_j^{\pm}(k) - \ef_j(b) ) \leq C_j b {\rm e}^{-\frac{k^2}{4b}},\ k \geq \kappa_j,
\ee
and the difference of the two eigenvalues is bounded as
\beq
0 \leq \omega_j^- (k) - \omega_j^+ (k) \leq 2 C_j b {\rm e}^{-\frac{k^2}{4b}},\ k \geq \kappa_j.
\eeq
\end{pr}

\begin{proof}

\noindent
1. Since $\Psi_j^{\rm HO}(k)$ is the eigenfunction of the harmonic oscillator Hamiltonian, we have for all $x \in \re$,
$$
(h(k)-\ef_j(b)) \Psi_j^{\rm HO}(\pm x,k) = ( (b x \pm k)^2 - (b x \mp k)^2 ) \chi_{\re_{\mp}}(x) \Psi_j^{\rm HO}(\pm x,k),
$$
so that for any $k \geq 0$, we have
\beq\label{eq:ho2}
 \| (h(k)-\ef_j(b)) \Psi_j^{\rm HO}(\pm x,k) \| \leq \| (bx \mp k)^2 \Psi_j^{\rm HO}(\pm x,k) \|_{{\rm L}^2(\re_{\mp})}.
\eeq
Here $\chi_I$ stands for the characteristic function of $I \subset \R$.
From \eqref{eq:ho2},
the identity
$$
\| (bx \mp k)^2 \Psi_j^{\rm HO}(\pm x,k) \|_{{\rm L}^2(\re_{\mp})}= \| (bx -k)^2 \Psi_j^{\rm HO}(x,k) \|_{{\rm L}^2(\re_-)},
$$
and \eqref{a2}, it follows that
\beq\label{eq:ho4}
\| (h(k)-\ef_j(b)) \Psi_j^{\rm HO}(\pm x,k) \| \leq c_j b {\rm e}^{-k^2 \slash (4b)},\ k \geq 0,
\eeq
for some constant $c_j>0$ depending only on $j$.

\noindent
2. Let $(\Psi_j^{\rm HO})^{\pm} (x,k):=(\Psi_j^{\rm HO}(x,k)\pm\Psi_j^{\rm HO}(-x,k)) \slash 2 \in \HH_{\pm}$.
In light of \eqref{eq:ho4} we have
\bel{a3}
\| (h(k)-\ef_j(b)) (\Psi_j^{\rm HO})^{\pm}(k) \|_{\HH} \leq c_j b {\rm e}^{-k^2 \slash (4b)},\ k \geq 0 .
\ee
Further since
$$
\| (\Psi_j^{\rm HO})^{\pm}(k) \|^2=\left( 1 \pm \int_{\re} \Psi_j^{\rm HO}(x,k) \Psi_j^{\rm HO}(-x,k) dx \right) \slash 2 ,
$$
with
$$
\left| \int_{\re} \Psi_j^{\rm HO}(x,k) \Psi_j^{\rm HO}(-x,k) dx \right| \leq \tilde{c}_j {\rm e}^{- k^2 \slash (4b)},
$$
for some constant $\tilde{c}_j>0$ depending only on $j$, we deduce from \eqref{a3} that
\bel{a3b}
{\rm dist}(\sigma(h^{\pm}(k)),\ef_j(b)) \leq C_j b {\rm e}^{-k^2 \slash (4b)},\ k \geq 0,
\ee
where $C_j>0$ depends only on $j$.

\noindent
3. As $\omega_j^-(k) > \ef_j(b)$ for each $k \in \re$, from the minimax principle,
the result \eqref{a0} for $\omega_j^-(k)$ follows readily from \eqref{a3b}. The case of $\omega^+$ is more complicated. In the section \ref{subsec:effective1}, we prove in the derivation of Proposition \ref{prop:effective-mass1} that the band function $\omega_j^+(k)$ has a unique absolute minimum at a value $\kappa_j \in (0, \ef_{2j-1}(b)^{1/2})$. Furthermore, $\omega_j^+(\kappa_j) \in ( \ef_{j-1}(b), \ef_j(b))$. We also prove that $(\omega_j^+)'(k) < 0$ for $k < \kappa_j$ and $(\omega_j^+)'(k) > 0$ for $k > \kappa_j$. The facts that the analytic band function is monotone increasing for $k > \kappa_j$ and converges to $\ef_j(b)$ as $k \rightarrow \infty$ due to \eqref{a3b} imply the result \eqref{a0} for $\omega_j^+(k)$.
\end{proof}


\subsection{Even band functions $\omega_j^+(k)$: the effective mass}\label{subsec:effective1}

We prove that the even states in $\mathcal{H}_+$, with band functions $\omega_j^+ (k) = \omega_{2j-1}(k)$,
have a unique positive minimum at $\kappa_j$. We prove that the even band function $\omega_j^+(k)$ is concave at $\kappa_j$.
This convexity means that there is a positive effective mass. This positive effective mass plays an important role in the perturbation theory and creation of the discrete spectrum discussed in section \ref{sec:ev-counting1}.

\begin{pr}\label{prop:effective-mass1}
The band functions $\omega_j^+(k) = \omega_{2j-1}(k)$, corresponding to the even states of $h(k)$, each have a unique extremum $\EE_j \in (\ef_{j-1}(b),\ef_j(b))$ that is a strict minimum. The minimum is attained at a single point $\kappa_j \in (0,\ef_{2j-1}(b)^{1 \slash 2})$. This point is the unique real solution of $\omega_{2j-1} (k) - k^2 = 0$, and $\EE_j = \kappa_j^2$. The concavity of the band function
at $\kappa_j$ is strictly positive and given by:
\beq\label{eq:eff-mass1}
(\omega_j^+)''(\kappa_j) = \omega_{2j-1}''(\kappa_j) = \frac{4 \kappa_j}{b} \psi_{2j-1} (0, \kappa_j)^2 > 0.
\eeq
We also have
$\pm (\omega_j^+)'(k)<0$ for $\pm (k -\kappa_j) <0$. 
\end{pr}

\begin{proof}
\noindent
1. We first prove that there exists a unique minimum for the band function. The Feynman-Hellmann formula yields
\bel{a5}
(\omega_j^+)'(k)=-2 \int_{\re} (b|x|-k) \psi_j^+(x,k)^2 dx,\ k \in \R.
\ee
Next, recalling \eqref{eq:disp-even1}, we get that
\bel{a6}
(\omega_j^+)'(k) = \frac{2}{b} f_j^+(k) \psi_j^+(0,k)^2,\ f_j^+(k):=k^2-\omega_j^+(k),
\ee
since $\psi_j^+(0,k) \neq 0$ and $(\psi_j^+)'(0,k) = 0$.
Moreover, taking into account that $h(0)=h_{\rm HO}(0)$ we see that
\bel{a6b}
\omega_j^+(k) \leq \omega_j^+(0) = \ef_{2j-1}(b),\ k \in \re^+,
\ee
as $h(k) \leq h_{\rm HO}(k)$ in this case. Therefore we have $f_j^+(0)=-\ef_{2j-1}(b)<0$ and
$f_j^+(k)>0$ for all $k > \ef_{2j-1}(b)^{1 \slash 2}$ from \eqref{a6b}.
The function $f_j^+$ is continuous in $\re$ hence there exists $\kappa_j \in (0, \ef_{2j-1}(b)^{1 \slash 2})$ such that $f_j^+(\kappa_j)=0$. Moreover, $f_j^+$ being real analytic, the set $\{ t \in \re,\ f_j^+(t)=0 \}$ is at most discrete so we may assume without loss of generality that $\kappa_j$ is its smallest element.

\noindent
2. We next prove that $\omega_j^+(k)$ is decreasing for $k < \kappa_j$ and increasing for $k > \kappa_j$.
It follows from \eqref{a5} that $(\omega_j^+)'(k)< 2k$. Integrating this inequality over the interval $[\kappa_j, k]$, we obtain
$$
\omega_j^+(k) < \omega_j^+(\kappa_j) + \int_{\kappa_j}^k 2 t dt = \omega_j^+(\kappa_j) + (k^2 - \kappa_j^2), ~~ k >\kappa_j,
$$
and hence $f_j^+(k) > f_j^+(\kappa_j)$ for all $k > \kappa_j$.
This result with the fact that $f_j^+ (\kappa_j) =0$ and \eqref{a6} imply that $(\omega_j^+)'(k) > 0$ for $k > \kappa_j$.

\noindent
3. To study the concavity of the band function and establish \eqref{eq:eff-mass1}, we differentiate \eqref{a6} with respect to $k$ and obtain
\beq\label{eq:eff-mass2}
(\omega_j^+)''(k)= -\frac{2}{b} \left( [(\omega_j^+)'(k)-2k ] \psi_j^+(0,k)^2 - 2 f_j^+(k) \psi_j^+(0,k) \partial_k \psi_j^+(0,k) \right),
\ \ k \in \re.
\eeq
We evaluate \eqref{eq:eff-mass2} at $\kappa_j$, recalling that $f_j^+(\kappa_j) = 0$ and that $(\omega_j^+)' (\kappa_j) = 0$,
in order to obtain \eqref{eq:eff-mass1}.

%
%

\noindent
4. We turn now to proving that $\EE_j(b) \in (\ef_{j-1}(b),\ef_j(b))$. Since $(\omega_j^+)'(k)>0$ for all $k \geq \kappa_j$ from Step 2 it follows readily from
\eqref{a0} that $\omega_j^+(\kappa_j) < \ef_j(b)$. Further it is clear that $\omega_1^+(\kappa_1)>0$ and we have in addition
$$\omega_j^+(k) =\omega_{2j-1}(k) > \omega_{2(j-1)}(k) = \omega_{j-1}^-(k) > \ef_{j-1}(b),\ k \in \R,\ j \geq 2, $$
so the result follows.
\end{proof}

%

In light of Proposition \ref{prop:effective-mass1}, we say that there exists an effective mass at $k=\kappa_1$, borrowing this term from the solid state physics. Further, it follows readily from Propositions \ref{prop:disp-deriv1} and \ref{prop:effective-mass1} that the spectrum of $H_0$ is absolutely continuous and equal to a half-line:
$$
\sigma(H_0) = \sigma_{\mathrm{ac}}(H_0) = [ \EE_1,+\infty).
$$


\subsection{Odd band functions $\omega_j^-(k)$: strict monotonicity}\label{subsec:even-monotone1}

The behavior of the odd band functions is much simpler.

\begin{pr}\label{prop:even-monotone1}
The odd band functions $\omega_j^-(k) = \omega_{2j}(k)$ are strictly monotone decreasing functions of $k \in \R$:
$$
(\omega_j^-)'(k)<0, ~~k \in \re .
$$
\end{pr}

\begin{proof}
Let us first recall from \eqref{eq:ev-deriv1} of Proposition \ref{prop:disp-deriv1} that for all $k \in \re$ we have the formula
\bel{a4}
(\omega_j^{\pm})'(k) = - \frac{2}{b} \left( (\omega_j^{\pm}(k)-k^2) \psi_j^{\pm}(0,k)^2 + (\psi_j^{\pm})'(0,k)^2 \right).
\ee
Bearing in mind that $\psi_k^-(0,k) = 0$ and $(\psi_k^-)'(0,k) \neq 0$, the result follows immediately from \eqref{a4}.
\end{proof}



\section{Mourre estimates, perturbations, and stability of the absolutely continuous spectrum}\label{sec:mourre-est1}
\setcounter{equation}{0}

In this section we study the spectrum of the operator $H_0$ and its perturbations using a Mourre estimate. For the unperturbed operator $H_0$, we prove a Mourre estimate using the fiber operator $h(k)$. This implies a lower bound on the velocity operator for certain states proving the existence of edge currents. We prove that this estimate is stable with respect to a class of perturbations.


\subsection{Mourre estimate for $H_0$}\label{subsec:mourre-est-unpert1}

For all $E \in \re$ and all $\delta>0$ we note $\Delta_E(\delta):=[E-(\delta \slash 2)b,E+(\delta \slash 2)b]$.

\begin{lemma}
\label{lm-a2}
Let $n \in \N^*$, $E \in (\ef_n(b),\EE_{n+1})$ and $d_n(E)$ be the distance between $E/b$ and the set $\{ \ef_n(1),\EE_{n+1}(1) \})$, i.e.
$$
d_n(E):= 
\max \left( (E \slash b) -\ef_n(1), \EE_{n+1}(1)-(E \slash b) \right).
$$
Then there exists a constant $\delta_0=\delta_0(E) \in (0,d_n(E))$, independent of $b$, satisfying
\bel{a10}
\omega_j^{-1}(\Delta_E(2\delta_0)) = \emptyset,\ j \geq 2n+1,
\ee
and
\bel{a11}
\omega_i^{-1}(\Delta_E(2\delta_0)) \cap \omega_j^{-1}(\Delta_E(2\delta_0))= \emptyset,\ 1 \leq i \neq j \leq 2n.
\ee
Moreover, for every $j \in \N_{2n}^*$ there is a constant $c_{n,j}=c_{n,j}(E)>0$, independent of $b$, such that we have
\bel{a12}
-\omega_j'(k) \geq c_{n,j} b^{1 \slash 2},\ k \in \omega_j^{-1}(\Delta_E(2\delta_0)).
\ee
\end{lemma}

\begin{proof}
1. First \eqref{a10} follows readily from Proposition \ref{prop:effective-mass1} and the fact that $\Delta_E(2\delta_0) \cap [\EE_{n+1}(b),+\infty) = \emptyset$ for all $\delta_0 \in (0,d_n(E))$ since $E+\delta_0 b< \EE_{n+1}(b)$.

\noindent
2. Next we notice that $h(k)$ is unitarily equivalent to the operator $b \check{h}(k \slash b^{1 \slash 2})$, where
$$\check{h}(q) := - \frac{d^2}{dt^2} + (|t|-q)^2,\ q \in \re, $$
is defined on the dense domain ${\rm C}_0^{\infty}(\re) \subset {\rm L}^2(\re)$. More precisely it holds true that
$\VV_b h(k) \VV_b^* = b \check{h}(k \slash b^{1 \slash 2})$, where
$$(\VV_b \psi)(x) := b^{-1 \slash 4} \psi(x \slash b^{1 \slash 2}),\ \psi \in {\rm L}^2(\re), $$
is easily seen to be a unitary transform in ${\rm L}^2(\re)$. As a consequence we have
\bel{a13}
\omega_j(k) = b \check{\omega}_j(k \slash b^{1 \slash 2}),\ k \in \re,\ j \in \N^*,
\ee
where $\{ \check{\omega}_j \}_{j=1}^{\infty}$ is the set of eigenvalues (arranged in increasing order) of $\check{h}(k)$.
Let $a_{n,j}$, $j \in \N_{2n-1}^*$, be the unique real number obeying $\check{\omega}_j(a_{n,j})=\ef_n(1)$, set $a_{n,2n}:=+\infty$, and denote by $\check{\omega}_j^{-1}$ the function inverse to $\check{\omega}_j : (-\infty,a_{n,j}) \to (\ef_n(1),+\infty)$.
As the interval $[(E \slash b)-\delta_0,(E \slash b)+\delta_0] \subset (e_n(1),\EE_{n+1}(1))$ it is in the domain of each function $\check{\omega}_j^{-1}$, $j \in \N_{2n}^*$, and we have
\bel{a12b}
\check{\omega}_j^{-1}([(E \slash b)-\delta_0,(E \slash b)+\delta_0]) = [\check{\omega}_j^{-1}((E \slash b)+\delta_0),\check{\omega}_j^{-1}((E \slash b)-\delta_0)],
\ee
by Propositions \ref{prop:effective-mass1} and \ref{prop:even-monotone1}.
Further, since $\check{\omega}_{j+1}^{-1}(E) > \check{\omega}_j^{-1}(E)$ for all $j \in \N_{2n-1}^*$, the functions $\check{\omega}_j^{-1}$ are continuous, and $2n-1$ is finite, then there is necessarily $\delta_0 \in (0,d_n(E))$ such that we have
$$\check{\omega}_{j+1}^{-1}((E \slash b)+\delta_0) > \check{\omega}_{j}^{-1}((E \slash b)-\delta_0),\ j \in \N_{2n-1}^*. $$
This and \eqref{a13}-\eqref{a12b} yields \eqref{a11}.

\noindent
3. Finally, taking into account that $\check{h}(q)$ coincides with $h(q)$ in the particular case where $b=1$, we deduce from
Propositions \ref{prop:effective-mass1} and \ref{prop:even-monotone1} for any $\Delta \subset (\ef_n(1),\EE_{n+1}(1))$ that
\bel{a14}
\inf_{q \in \check{\omega}_j^{-1}(\Delta)} (-\check{\omega}_j'(q)) = \check{c}_j(\Delta) >0,\ j \in \N_{2n}^*,
\ee
where the constant $\check{c}_j(\Delta)$ is independent of $b$.
Now \eqref{a12} follows readily from \eqref{a14} since $[(E\slash b)-\delta_0 ,(E\slash b)+\delta_0)] \subset ( \ef(1),\EE_{n+1}(1))$ and
$$ \inf_{k \in \omega_j^{-1}(\Delta_E(2\delta_0))} (-\omega_j'(k))
=b^{1 \slash 2} \inf_{q \in \check{\omega}_j^{-1}([(E\slash b)-\delta_0 ,(E\slash b)+\delta_0)] )} (-\check{\omega}_j'(q)), $$
according to \eqref{a13}.
\end{proof}

Let us now introduce the operator $A=A^* := - y$ defined originally on ${\rm C}_0^{\infty}(\re^2)$. The operator $A$ extends to a self-adjoint operator in ${\rm L}^2(\re^2)$. Note that ${\rm C}_0^{\infty}(\re^2)$ is dense in $\Dom(H_0)$ and hence that $\Dom(A) \cap \Dom(H_0)$ is dense in $\Dom(H_0)$.

\begin{pr}
\label{prop-me1}
Let $b>0$, $n \in \N^*$, $E \in (\ef_n(b),\EE_{n+1}(b))$ and assume that $\delta_0 \in (0,d_n(E))$ is chosen to satisfy \eqref{a10}-\eqref{a11} according to Lemma \ref{lm-a2}.
Let $\chi \in {\rm C}_0^{\infty}(\re)$ with $\supp \chi \subset \Delta_E(2\delta_0)$.
Then there exists a constant $c_n=c_n(E)>0$, independent of $b$, such that we have
\bel{me1}
\chi(H_0) [ H_0 , i A ] \chi(H_0) \geq c_n b^{1 \slash 2} \chi(H_0)^2,
\ee
as a quadratic form on $\Dom(A) \cap \Dom(H_0)$.
\end{pr}

\begin{proof}
We get
\bel{me1b}
[H_0,i A]=-2(p_y - b |x|),
\ee
on $\Dom(A) \cap \Dom(H_0)$. We recall the orthogonal projection $P_j(k)$
defined by $P_j(k):=\langle . , \psi_j(k) \rangle \psi(k)$, for all $j \in \N^*$. The commutator on the left in \eqref{me1b}
fibers over $k \in \R$, so by a direct calculation, we find that
$$
\chi(H_0) [ H_0, i A] \chi(H_0) = -2 \FF^* \left( \sum_{j,m \in \N^*} \int_{\re}^{\oplus} \chi(\omega_j(k))  \chi(\omega_m(k)) P_j(k) (k-b |x|) P_m(k) dk \right) \FF.
$$
Taking into account that $\supp \chi \subset \Delta_E(2\delta_0)$, we deduce from \eqref{a10}-\eqref{a11} that
$$ \chi(H_0) [ H_0, i A] \chi(H_0) = -2 \FF^* \left( \sum_{j=1}^{2n} \int_{\re}^{\oplus} \chi(\omega_j(k))^2  \langle \psi_j(k) , (k-b |x|) \psi_j(k) \rangle P_j(k) dk \right) \FF, $$
whence
\bel{me2}
\chi(H_0) [ H_0, i A] \chi(H_0) = \FF^* \left( \sum_{j=1}^{2n} \int_{\re}^{\oplus} \chi(\omega_j(k))^2  (-\omega_j'(k)) P_j(k) dk \right) \FF,
\ee
from the Feynman-Hellmann formula. In light of \eqref{a12}, we have
$$ -\omega_j'(k) \chi(\omega_j(k))^2 \geq c_{n,j} b^{1 \slash 2} \chi(\omega_j(k))^2,\ j \in \N_{2n}^*, $$
so \eqref{me2} yields
$$ \chi(H_0) [ H_0, i A] \chi(H_0) \geq c_n b^{1 \slash 2} \FF^* \left( \sum_{j=1}^{2n} \int_{\re}^{\oplus} \chi(\omega_j(k))^2  P_j(k) dk \right) \FF = c_n b^{1 \slash 2} \chi(H_0)^2, $$
where $c_n:=\min_{j \in \N_{2n}^*} c_{n,j}>0$.
\end{proof}

Let $\proj_0(I)$ denote the spectral projection of $H_0$ for the Borel set $I \subset \re$. Then by choosing $\chi$ in Proposition \ref{prop-me1} to be equal to one on $\Delta_E(\delta_0)$ and multiplying \eqref{me1}
from both sides by $\proj_0(\Delta_E(\delta_0))$, we obtain the following Mourre estimate for $H_0$:
\bel{me2b}
\proj_0(\Delta_E(\delta_0)) [H_0,iA] \proj_0(\Delta_E(\delta_0)) \geq c_n b^{1 \slash 2} \proj_0(\Delta_E(\delta_0)).
\ee

\subsection{Edge currents for $H_0$}\label{subsec:edge1}

We can prove the existence of edge currents for the unperturbed Hamiltonian $H_0$ based on the Mourre estimate \eqref{me2b}.
A state $\varphi \in {\rm L}^2 (\R^2)$ carries an edge current of the Hamiltonian $H$ if $J_y (\varphi) := \langle \varphi, v_y \varphi \rangle$ is strictly positive, where the velocity operator is $v_y = (i/2)[H,A]$.

\begin{follow}\label{corollary-edge-currents1}
Let $b$, $n$, $E$, and $\delta_0$ be as in Proposition \ref{prop-me1}.
Let $\varphi \in {\rm L}^2(\re^2)$ satisfy $\varphi=\mathbb{P}_0(\Delta_E(\delta_0)) \varphi$. Then $\varphi$ carries an edge current and the edge current is bounded below by
\beq\label{eq:edge-current-lb1}
J_y(\varphi) \geq \frac{c_n}{2} b^{1/2} \| \varphi \|^2,
\eeq
where $c_n$ is the constant defined in Proposition \ref{prop-me1}.
\end{follow}

The proof of this corollary follows directly from \eqref{me2b} since for $\varphi$ as in the corollary, we have $J_y (\varphi)  =
\langle \varphi,  (1/2) \mathbb{P}_0(\Delta_E(\delta_0)) [H_0, iA] \mathbb{P}_0(\Delta_E(\delta_0)) \varphi \rangle$.
The edge currents associated with $H_0$ and states $\varphi$ as in Corollary \ref{corollary-edge-currents1} are also localized in a neighborhood of size roughly $b^{-1/2}$ about $x=0$. This follows from Proposition \ref{pr-localcurrent1}.


\subsection{Stability of the Mourre estimate}\label{subsec:mourre-stability1}

One of the main benefits of a local commutator estimate like \eqref{me2b} is its stability under perturbation. Namely we consider the perturbation of $H_0=(-i \nabla -A_0)$,
$A_0=A_0(x,y):=(0,b|x|)$, by a magnetic potential $a(x,y)=(a_1(x,y),a_2(x,y)) \in {\rm W}^{1,\infty}(\re^2)$ and a bounded scalar potential $q(x,y) \in {\rm L}^{\infty}(\re^2)$. We
prove that a Mourre inequality for the perturbed operator
\bel{me2c}
H=H(a,q):= (- i\nabla - A_0 - a)^2 + q = (p_x - a_1)^2 +(p_y - b |x| - a_2)^2 + q,
\ee
remains true provided $\| a \|_{W^{1,\infty}(\re^2)}$ and $\| q \|_\infty$ are small enough relative to $b$.
We preliminarily notice that
\bel{me2d}
W = W(a):= H(a,0) - H_0=2 a \cdot ( -i \nabla - A_0) - i (\nabla \cdot a ) + a \cdot a,
\ee
with $\| (-i \nabla - A_0) \varphi \| = \langle H_0 \varphi , \varphi \rangle^{1 \slash 2} \leq \lambda \| H_0 \varphi \| + \lambda^{-1} \| \varphi \|^2$ for all $\varphi \in C_0^{\infty}(\re^2)$ and $\lambda >0$, so we have
$$ \| W \varphi \| \leq 2 \lambda \| a \|_{\infty} \| H_0 \varphi \| + (\lambda^{-1} + \| \nabla a \|_\infty + \| a \|_\infty^2) \| \varphi \|,\ \lambda>0. $$
Taking $\lambda=1 \slash (4 \| a \|_{\infty})$ in the above inequality we find that $W$ is $H_0$-bounded with relative bound smaller than one. In light of \cite{RS-IV}[Theorem X.12] the operator $H(a,0)=H_0+W$
is thus selfadjoint in ${\rm L}^2(\re^2)$ with same domain as $H_0$, and the same is true for $H=H(a,q)=H(a,0)+q$ since $q \in {\rm L}^{\infty}(\re^2)$.

\begin{pr}
\label{prop-me2}
Let $b$, $n$, $E$ and $\delta_0$ be as in Proposition \ref{prop-me1}. Assume that $\delta=\delta(E) \in (0,\delta_0)$, $a \in
{\rm W}^{1,\infty}(\re^2)$ and $q \in L^\infty(\re^2)$ verify
\bel{me3b}
F_{n,E} \left( \delta, \frac{\| q \|_{\infty}}{b}, \frac{\| a \|_{\infty}^2 + \| \nabla a \|_{\infty}}{b} \right) < \frac{1}{2},
\ee
\footnote{Notice that the function $F_{n,E}$ depends on $E$ through $\delta_0=\delta_0(E)$.}{where}
\bel{me3c}
F_{n,E}(\delta,\mathfrak{a},\mathfrak{q}):= \left( \frac{f_n(\delta,\mathfrak{a},\mathfrak{q})}{\delta_0} \right)^2 + \frac{2}{c_n} \left( \mathfrak{a}^{1 \slash 2} + (2n+1+f_n(\delta,\mathfrak{a},\mathfrak{q}))^{1 \slash 2} \left( \frac{f_n(\delta,\mathfrak{a},\mathfrak{q})}{\delta_0} \right)^{1 \slash 2} \right),
\ee
$f_n$ is given by \eqref{me7e} and $c_n$ is the constant defined in Proposition \ref{prop-me1}.
Then we have the following Mourre estimate
\bel{me3}
\proj(\Delta_E(\delta)) [H,iA] \proj(\Delta_E(\delta)) \geq \frac{c_n}{2} b^{1 \slash 2} \proj(\Delta_E(\delta)),
\ee
where $\proj(I)$ denotes the spectral projection of $H$ for the Borel set $I \subset \re$.
\end{pr}
\begin{proof}
By combining the following decomposition of $\psi \in  \proj(\Delta_E(\delta)) {\rm L}^2(\re^2)$ into the sum
\bel{me2g}
\psi = \phi + \xi,\ \phi:=\proj_0(\Delta_E(\delta_0)) \psi,\ \xi:=\proj_0(\re \setminus \Delta_E(\delta_0)) \psi,
\ee
with the basic equality
\bel{me2h}
[H , iA]=[H_0,iA] + 2 a_2,
\ee
obtained through standard computations, we get that
$$
\langle \psi , [H , iA] \psi \rangle
=\langle \phi , [H_0 , iA] \phi \rangle + 2 \langle \psi , a_2 \psi \rangle + C(\phi,\xi),
$$
with
\beas
 C(\phi,\xi) & := & \int_{\re} \langle \hat{\xi}(\cdot,k) , (k-b|x|) \hat{\xi}(\cdot,k) \rangle_{{\rm L}^2(\re)} dk \\
& & + 2 \Pre{\int_{\re} \langle \hat{\phi}(\cdot,k) , (k-b|x|) \hat{\xi}(\cdot,k) \rangle_{{\rm L}^2(\re)} dk}.
\eeas
This entails
\bel{me6}
\langle \psi , [H , iA] \psi \rangle
\geq \langle \phi , [H_0,iA] \phi \rangle - 2 \left( \| a \|_\infty \| \psi \| + \| (p_y - b |x|) \xi \| \right) \| \psi \|,
\ee
since
$$
\| \psi \|^2 = \| \phi \|^2 + \| \xi \|^2,
$$
as can be seen from the orthogonality of $\phi$ and $\xi$ in ${\rm L}^2(\R^2)$, arising from \eqref{me2g}.
The first term in the r.h.s of \eqref{me6} is lower bounded by \eqref{me2b} as
\bel{me5b}
\langle \phi , [H_0,iA] \phi \rangle \geq c_n b^{1 \slash 2} \| \phi \|^2,
\ee
and $\| (p_y - b |x|) \xi \|$ can be majorized with the help of the estimate
$$ \|  (p_y - b |x|) \xi \|^2  \leq \langle \xi , H_0 \xi \rangle = \langle \psi , H_0 \xi \rangle = \langle H_0 \psi , \xi \rangle=\langle (H-W-q) \psi , \xi \rangle, $$
giving
\bel{me6b}
\|  (p_y - b |x|) \xi \|^2 \leq \langle \xi , H_0 \xi \rangle \leq \left( (E \slash b)+ \delta + \mathfrak{q}  + w \right) b \| \xi \| \| \psi \|,
\ee
where $\mathfrak{q}:= \| q \|_\infty \slash b$ and $w:=\| W \psi \| \slash (b \| \psi \|)$.
Further we have
\bel{me7}
\| \xi \| \leq \frac{\delta + \mathfrak{q}  + w}{\delta_0} \| \psi \|,
\ee
since $\| \xi \|^2=\langle (H-E-W-q) \psi , (H_0-E)^{-1} \xi \rangle$. In light of the r.h.s in \eqref{me6b}-\eqref{me7} we are thus left with the task of majorizing $w$.
This can be done by combining the estimate
$$ \| (-i \nabla - A_0) \psi \|  = \langle H_0 \psi , \psi \rangle^{1 \slash 2} \leq \| H_0 \psi \|^{1 \slash 2} \|\psi \|^{1 \slash 2} \leq
\| (H-W-q) \psi \|^{1 \slash 2} \|\psi \|^{1 \slash 2}, $$
entailing $\| (-i \nabla - A_0) \psi \| \leq \left( (E \slash b)+\delta+\mathfrak{q}+w \right)^{1 \slash 2} b^{1 \slash 2} \| \psi \|$,
with \eqref{me2d}. We find out that
$w \leq \mathfrak{a}^{1 \slash 2} \left(  \mathfrak{a}^{1 \slash 2}+ 2 \left( (E \slash b) +\delta+\mathfrak{q}+w \right)^{1 \slash 2} \right)$ with $\mathfrak{a}:= (\| a \|_\infty^2 + \| \nabla a \|_{\infty}) \slash b$, whence
$w  \leq  2 \mathfrak{a}^{1 \slash 2} \left( 3 \mathfrak{a}^{1 \slash 2} + (E \slash b) + \delta + \mathfrak{q} \right)^{1 \slash 2} )$. From this, the estimate $E \leq (2n+1) b$, arising from Proposition \ref{prop:effective-mass1}, and \eqref{me6b}-\eqref{me7} then follows that
\bel{me7c}
\| \xi \| \leq \frac{f_n(\delta,\mathfrak{a},\mathfrak{q})}{\delta_0} \| \psi \|,
\ee
and
\bel{me7d}
\| (p_y-b|x|) \xi \| \leq (2n+1 + f_n(\delta,\mathfrak{a},\mathfrak{q}))^{1 \slash 2} \left( \frac{f_n(\delta,\mathfrak{a},\mathfrak{q})}{\delta_0} \right)^{1 \slash 2} b^{1 \slash 2} \| \psi \|,
\ee
where
\bel{me7e}
f_n(\delta,\mathfrak{a},\mathfrak{q}):=\delta + \mathfrak{q} + 2 \mathfrak{a}^{1 \slash 2} \left( 3 \mathfrak{a}^{1 \slash 2}+ (2n+1+\delta + \mathfrak{q})^{1 \slash 2} \right).
\ee
Putting \eqref{me6}-\eqref{me5b} and \eqref{me7c}-\eqref{me7d} together and recalling \eqref{me3c} we end up getting that
$$
\langle \psi , [H , iA] \psi \rangle \geq c_n \left( 1-F_{n,E}(\delta, \mathfrak{q},\mathfrak{a}) \right) b^{1 \slash 2} \| \psi \|^2,
$$
so \eqref{me3} follows readily from this and \eqref{me3b}.
\end{proof}


\subsection{Absolutely continuous spectrum}\label{subsec:acspectrum1}

We now apply Proposition \ref{prop-me2} to prove the existence of absolutely continuous spectrum for perturbed magnetic barrier operators.
Using direct computation, we deduce from \eqref{me1b} and \eqref{me2h} that $[[H,iA],iA]=-2$. Hence the double commutator of $H$ with $A$ is
bounded from $\Dom(H)=\Dom(H_0)$ to ${\rm L}^2(\re^2)$. Moreover, since $[H,iA]$ extends to a bounded operator from $\Dom(H_0)$ to ${\rm L}^2(\re^2)$, the Mourre estimate \eqref{me3} combined with \cite{CFKS}[Corollary 4.10] entails the following:

\begin{follow}
\label{cor-me3}
Let $b$, $n$, $E$ and $\delta_0$ be the same as in Proposition \ref{prop-me1}. Assume that $\delta \in (0,\delta_0)$, $q \in {\rm L}^\infty(\re^2)$ and
$a \in {\rm W}^{1,\infty}(\re^2)$ satisfy \eqref{me3b}.
Then the spectrum of $H=H(a,q)$ in $\Delta_E(\delta)$ is absolutely continuous.
\end{follow}

Armed with Corollary \ref{cor-me3} we turn now to proving the main result of this section.

\begin{thm}
\label{thm-me}
Let $b>0$, $n \in \N^*$, and let $\Delta$ be a compact subinterval of $(\ef_n(b),\EE_{n+1}(b))$. Then there are two constants $\mathfrak{a}^*=\mathfrak{a}^*(n,\Delta)>0$ and $\mathfrak{q}^*=\mathfrak{q}^*(n,\Delta)>0$, both independent of $b$, such that for all $(a,q) \in {\rm W}^{1,\infty}(\re^2) \times {\rm L}^2(\re^2)$ verifying
$\| a \|_\infty^2 + \| \nabla a \|_\infty \leq \mathfrak{a}^* b$ and  $\| q \|_{\infty} \leq \mathfrak{q}^* b$,
the spectrum of $H=H(a,q)$ in $\Delta$ is absolutely continuous.
\end{thm}
\begin{proof}
For every $E \in \Delta$ choose $\delta(E) \in (0,\delta_0(E))$, $\mathfrak{a}(E)>0$ and $\mathfrak{q}(E)>0$ such that
\bel{me7b}
F_{n,E}(\delta(E),\mathfrak{a}(E),\mathfrak{q}(E)) < \frac{1}{2},
\ee
where $F_{n,E}$ is defined in \eqref{me3c}.

Since $\Delta$ is compact and $\Delta \subset \cup_{E \in \Delta} \Delta_E(\delta(E))$, there exists a finite set $\{ E_j \}_{j=1}^N$ of energies in $\Delta$ such that
\bel{me8}
\Delta \subset \bigcup_{j= 1}^N \Delta_{E_j}(\delta(E_j)).
\ee
Set $\mathfrak{a}^*:=\min_{1 \leq j \leq N} \mathfrak{a}(E_j)>0$ and $\mathfrak{q}^*:=\min_{1 \leq j \leq N} \mathfrak{q}(E_j)>0$.
Since $F_{n,E_j}(\delta(E_j),\cdot,\cdot)$, $j=1,\ldots,N$, is an increasing function of each of the two last variables taken separately, when the remaining one is fixed, we necessarily have $F_{n,E_j}(\delta(E_j),\mathfrak{a}^*,\mathfrak{q}^*) < 1 \slash 2$ by \eqref{me7b}. Assume that $\| a \|_\infty^2 + \| \nabla a \|_\infty \in [0,\mathfrak{a}^*b)$ and $\| q \|_{\infty} \in [0,\mathfrak{q}^*b)$.
For every $j=1,\ldots,N$, the spectrum of $H$ in $\Delta_{E_j}(\delta(E_j))$ is thus absolutely continuous by Corollary \ref{cor-me3} so the result follows from this and \eqref{me8}.
\end{proof}


\section{Edge currents: existence, stability, localization, and asymptotic velocity}\label{sec:edge-curr1}
\setcounter{equation}{0}

A major consequence of the Mourre estimate in Proposition \ref{prop-me1} for the unperturbed operator $H_0$ is the lower bound on the edge current carried by certain states given in Corollary \ref{corollary-edge-currents1}. Because of the stability result for the Mourre estimate for the perturbed operator $H(a,q)$ in Proposition \ref{prop-me2}, we prove in this section that edge currents are stable under perturbations.
We then prove that these currents are well-localized in a strip of width $\mathcal{O}(b^{-1/2})$ about $x=0$.
Finally, we prove that the asymptotic velocity is bounded from below demonstrating that the edge currents persist for all time.


\subsection{Existence and stability of edge currents}\label{subsec:edge-currents1}

For the perturbed operator $H=H(a,q)$, the $y$-component of the velocity operator is
\bel{ec1}
v_{y,a,q} := (1 \slash 2) [ H, iA ] = -(p_y - b |x|)+a_2, ~~A = -y,
\ee
according to \eqref{me1b} and \eqref{me2h}.
A state $\varphi \in {\rm L}^2 (\R^2)$ carries an edge current if
\bel{ec2}
J_{y,a,q} ( \varphi) := \langle \varphi, v_{y,a,q} \varphi \rangle \geq c \| \varphi \|^2,
\ee
for some constant $c> 0$.
For notational simplicity we write $v_y$ (resp. $J_y$) instead of $v_{y,0,0}$ (resp. $J_{y,0,0}$) in the particular case of the unperturbed operator $H_0$ corresponding to $a=0$ and $q=0$. We consider states in the range of the spectral projector $\P_0 (\cdot)$ for $H_0$,
and in the range of the spectral projector $\P_{a,q}(\cdot)$ for $H=H(a,q)$, and energy intervals as in \eqref{me2b} for $H_0$, and in \eqref{me3} for $H_{a,q}$. We then deduce from \eqref{ec1}-\eqref{ec2} the existence of edge currents for the operator $H_0$ and $H_{a,q}$, respectively. We recall Corollary \ref{corollary-edge-currents1} in the first part of the following theorem.

\begin{thm}
\label{lm-a3}
Let $b$, $n$, $E$, and $\delta_0$ be as in Proposition \ref{prop-me1}.
\begin{enumerate}
\item Let $\varphi \in {\rm L}^2(\re^2)$ satisfy $\varphi=\mathbb{P}_0(\Delta_E(\delta_0)) \varphi$. Then $\varphi$ carries an edge current obeying
$$
J_y(\varphi) \geq \frac{c_n}{2} b^{1/2} \| \varphi \|^2,
$$
where $c_n$ is the constant defined in Proposition \ref{prop-me1}.
\item Let $\delta \in (0,\delta_0)$ and assume that $(a,q) \in W^{1,\infty}(\R^2) \times {\rm L}^\infty(\R^2)$
verifies the condition \eqref{me3b}, where $(c_n \slash 2)$ is substituted for $c_n$ in the definition \eqref{me3c}.
Then every state $\varphi \in \mathbb{P}_{a,q}(\Delta_E(\delta)) {\rm L}^2(\re^2)$ carries an edge current and we have the lower bound
\beq\label{eq:edge-current-lbpert1}
J_{y,a,q}(\varphi) \geq \frac{c_n}{4} b^{1/2} \| \varphi \|^2.
\eeq
\end{enumerate}
\end{thm}

\subsection{Localization of edge currents}\label{subsec:edge-local1}

We establish the localization of the edge currents described in Theorem \ref{lm-a3} using a method introduced by Iwatsuka \cite[section 3]{iwatsuka1}. We refer the reader to section \ref{subsec:mourre-est-unpert1} for the definitions of the various quantities appearing in the following proposition.

\begin{pr}
\label{pr-localcurrent1}
Let $n$, $E$, and $\Delta_0$ be as in Proposition \ref{prop-me1} and
choose $\delta = \delta(E) \in (0, \delta_0)$ in accordance with condition \eqref{me3b}. Let $\varphi \in {\rm L}^2(\re^2)$ satisfy $\varphi=\mathbb{P}_0(\Delta_E(\delta)) \varphi$ with $\| \varphi \| = 1$.
Then for all $\eps>0$ there
exists a constant $\tilde{b}>0$, depending only on $n$, $\delta$ and $\eps$ such that we have
$$
\int_{\re^2} \chi_{I_\eps}(x) | \varphi(x,y) |^2 {\rm d} x {\rm d} y
\geq  1 - \sqrt{2} {\rm e}^{-b^\eps},
$$
for $b \geq \tilde{b}$. Here $\chi_{I_{\eps}}$ is the characteristic function
of the interval $I_{\eps}:= [-b^{-1 \slash 2+\eps},b^{-1 \slash 2 + \eps}]$.
\end{pr}

\begin{proof}
1.\ Due to \eqref{a0} we have
$$ \max \{ \sup \omega_j^{-1}(\Delta_E(\delta)),\ j=1,\ldots,2n \} \leq \alpha_n b^{1 \slash 2}, $$
for some constant $\alpha_n>0$, depending only on $n$ and $\delta$.
Hence there is a constant $\beta_n>0$, depending only on $n$ and $\delta$, such that the estimate
\bel{ec3}
Q_j(x,k):=V_{eff}(x,k)-\omega_j(k) \geq b^2(|x|- x_n)^2>0,
\ee
holds for all $j=1,\ldots,2n$, $k \in \omega_j^{-1}(\Delta_E(\delta))$ and $| x | \geq x_n:=\beta_n b^{1 \slash 2}$.\\
2.\ We will prove that an eigenfunction $\psi_j(k)$, for $k \in \omega_j^{-1}(\Delta_E(\delta))$,
decays in the region $|x| \geq  x_n$. In particular, we will establish for $j=1,\ldots,2n$ that
\bel{a30}
|\psi_j(x,k)| \leq  \left( \frac{2 b}{\pi} \right)^{1 \slash 4} {\rm e}^{-b(|x|-x_n)^2 \slash 2},\  |x| \geq  x_n,\ k \in \omega_j^{-1}(\Delta_E(\delta)).
\ee
Let $j \in \N_{2n}^*$ and $k \in \omega_j^{-1}(\Delta_E(\delta))$ be fixed.
In light of \eqref{ec3} and the differential equation $\psi_j(x,k)''=Q_j (x,k) \psi_j(x,k)$ we have $\psi_j(x,k) \psi_j'(x,k)<0$ for $|x| > x_n$, by \cite{iwatsuka1}[Proposition 3.1].
This implies that
\beq\label{eq:product1}
\frac{\psi_j^\prime(x,k)}{\psi_j(x,k)} = \frac{\psi_j^\prime(x,k) \psi_j(x,k)}{\psi_j(x,k)^2} < 0,\ x > x_n.
\eeq
Following \cite{iwatsuka1}[Lemma 3.5], differentiating $I(x,k) :=\psi_j'(x,k)^2 - Q_j(x,k) \psi_j(x,k)^2$, one finds that $\partial_x I(x,k) < 0$ since $Q_j'(x,k) > 0$ in the region $x > x_n$.
Since $I(x,k)$ vanishes at infinity, due to the vanishing of $\psi_j(x,k)$ and $\psi_j^\prime(x,k)$ established by \cite{iwatsuka1}[Lemma 3.3], this means that $I(x,k) > 0$ in the region $x > x_n$. From this we conclude that
\beq\label{eq:log-ineq1}
\psi_j^\prime(x,k)^2 \geq Q_j(x,k) \psi_j(x,k)^2,\ x > x_n.
\eeq
As a consequence of \eqref{ec3} and \eqref{eq:product1}-\eqref{eq:log-ineq1}, we find that
$$
\frac{\psi_j' (x,k)}{\psi_j(x,k)}  \leq - \sqrt{ Q_j (x,k)} \leq - b(x-x_n), ~~{\rm for} ~~ x > x_n .
$$
Result \eqref{a30} follows from integrating this differential inequality over the region $x > x_n$ and arguing in the same way as above for $x < -x_n$.

\noindent
3. Choose $b$ so large that $b^\eps > (1+\beta_n^{1 \slash 2})^2$. Then we have $b^{-1 \slash 2 + \eps} > x_n + b^{(-1+  \eps) \slash 2}$ by elementary computations, whence
\bel{a32}
\int_{\R \setminus I_\eps} \psi_j(x,k)^2 dx \leq 2 \left( \frac{2 b}{\pi} \right)^{1 \slash 2} \int_{b^{-1 \slash 2 + \eps}}^{+\infty}
{\rm e}^{-b(x-x_n)^2} dx \leq \sqrt{2} e^{-b^\eps},
\ee
from \eqref{a30}.
Finally, since
\beas
\int_{\re^2} \chi_{I_\eps}(x) | \varphi(x,y) |^2 {\rm d} x {\rm d} y &
= & \int_{\re^2}  \chi_{I_\eps}(x) | \hat{\varphi}(x,k) |^2 {\rm d} x {\rm d} k \\
& = & \sum_{j=1}^{2n} \int_{\omega_j^{-1}(\Delta_E(\delta))} | \beta_j(k) |^2 \left( \int_{I_\eps} \psi_j(x,k)^2 {\rm d} x \right) {\rm d} k,
\eeas
by Lemma \ref{lm-a2}, where $\beta_j(k):= \langle \hat{\varphi}(\cdot,k) , \psi_j(\cdot,k) \rangle$, the result follows readily from \eqref{a32} and the identity $\sum_{j=1}^{2n} \int_{\omega_j^{-1}(\Delta_E(\delta))} | \beta_j(k)|^2 {\rm d} k =1$.
\end{proof}


\subsection{Persistence of edge currents in time: Asymptotic velocity}\label{subsec:time-behavior1}

We investigate the time evolution of the edge current under the unitary evolution groups generated by the
Iwatsuka Hamiltonians $H_0$ \eqref{eq:basic1}, and by the perturbed Iwatsuka Hamiltonians $H(a,q)$ \eqref{me2c}.
The general situation we address is the following. Let $H$ be a self-adjoint Schr\"odinger operator on ${\rm L}^2 (\R^2)$.
This operator generates the unitary time evolution group $U(t) := e^{-itH}$. Let $v_y := (i/2) [ H, A]$, with $A = -y$,
        be the $y$-component of the velocity operator. We are interested in evaluating the asymptotic time behavior of $\langle U(t) \varphi, v_y U(t) \varphi \rangle$ as
$t \rightarrow \pm \infty$ for appropriate functions $\varphi$.

The lower bounds on the edge currents for the unperturbed and the perturbed Iwatsuka models are valid for all times.
It we replace $v_y$ in the expression $J_y (\varphi) = \langle \varphi, v_y \varphi \rangle$ in Corollary \ref{corollary-edge-currents1}
by $v_y (t) := e^{itH_0}v_ye^{-itH_0}$, then the lower bound \eqref{eq:edge-current-lb1}
remains valid since the state $\varphi (t) := e^{-it H_0} \varphi$ satisfies $\varphi(t) \in \mathbb{P}_0(\Delta_E(\delta_0)) {\rm L}^2 (\R^2)$
for all time.
Similarly, if we replace $v_{y,a,q}$ in \eqref{ec2} by its time evolved current $v_{y,a,q}(t)$
using the operator $e^{-itH(a,q)}$, then the lower bound in \eqref{eq:edge-current-lbpert1} remains valid for all time.


Perturbed Hamiltonians $H(a,q)$ were treated in sections \ref{sec:mourre-est1} and \ref{sec:edge-curr1}. Part 2 of Theorem \ref{lm-a3} states that if the ${\rm L}^\infty$-norms of $a_j^2$, $\nabla a_j$, for $j=1,2$, and of $q$ are small relative to $b$ in the sense that condition \eqref{me3b} is satisfied, then the edge current $J_{y,a,q}(\psi)$ is bounded from below for all $\psi \in \mathbb{P}_{a,q} (\Delta_E (\delta)) {\rm L}^2(\R^2)$, where $\Delta_E (\delta)$ is defined at the beginning of section \ref{sec:mourre-est1} and $(E, \delta)$ are as in Proposition \ref{prop-me1}
and Proposition \ref{prop-me2}. This relative boundedness of $a_j$, $\nabla a_j$, and of $q$ is rather restrictive. From the form of the current operator in \eqref{ec1}, it would appear that only $\|a_2\|_\infty$ needs to be controlled. We prove here that if we limit the support of the perturbation $(a_1, a_2, q)$ to a strip of arbitrary width $R$ in the $y$-direction, and require only that $\| a_2 \|_\infty$ be small relative to $b^{1/2}$,
then the {\it asymptotic velocity} associated with energy intervals $\Delta_E (\delta)$ and the perturbed Hamiltonian $H(a,q)$ exists and satisfies the same lower bound as in \eqref{eq:edge-current-lbpert1}. Furthermore, the spectrum in $\Delta_E (\delta)$ is absolutely continuous. This means that the edge current is stable with respect to a different class of perturbations than in Theorem \ref{lm-a3}.

We recall that the \emph{asymptotic velocity} associated with a pair of self-adjoint operators $(H_0, H_1)$ is defined in terms of the local wave operators for the pair, see, for example \cite[section 4.5--4.6]{DG}. The local wave operators $\Omega_\pm (\Delta)$ for an energy interval $\Delta \subset \R$ are defined as the strong limits:
\beq\label{eq:local-wo1}
\Omega_\pm (\Delta ) := s-\lim_{t \rightarrow \pm \infty} e^{i t H_1} e^{-i t H_0} \mathbb{P}_{0,ac} (\Delta),
\eeq
where $\mathbb{P}_{0,ac} (\Delta)$ is the spectral projector for the absolutely continuous subspace of $H_0$ associated with the interval $\Delta$.
For any $\varphi$, we define the \emph{asymptotic velocity} $V_y^\pm (\Delta)$ of the state $\varphi$ by
\beq\label{eq:av-defn1}
\langle \varphi,V_y^\pm (\Delta) \varphi \rangle  := \langle \varphi, \Omega_\pm (\Delta) v_y \Omega_\pm (\Delta)^* \varphi \rangle .
\eeq
In the case that $H_0$ commutes with $v_y$, it is easily seen from the definition \eqref{eq:local-wo1} that
$$
\langle \varphi,V_y^\pm (\Delta) \varphi \rangle  = \lim_{t \rightarrow \pm \infty} \langle \varphi, e^{itH_1} \mathbb{P}_{0,ac} (\Delta) v_y
\mathbb{P}_{0,ac} (\Delta) e^{-itH_1}  \varphi \rangle.
$$

Our main result is the existence of the asymptotic velocity \eqref{eq:av-defn1}
in the $y$-direction for the perturbed operators $H(a,q)$ described in section \ref{subsec:mourre-stability1}. We prove that the asymptotic velocity satisfies the lower bound given in \eqref{eq:asympt-velocity1} provided the perturbations $(a,q)$
have compact support in the $y$-direction. The local wave operators appearing in the definition \eqref{eq:av-defn1} are constructed from
the pair $(H_0, H_1)$ where $H_0$ is the unperturbed Iwatsuka Hamiltonian and $H_1 = H(a,q)$.
As discussed in section \ref{subsec:ac-spectrum1}, the spectrum of $H_0$ is purely absolutely continuous.

\begin{thm}\label{th:asymptotic1}
Let $b$, $n$, $E$, and $\delta_0$, be as in Proposition \ref{prop-me1} and for any $0 < \delta \leq \delta_0$, let
$\Delta_E (\delta)$ be as defined in section \ref{subsec:mourre-est-unpert1}. Suppose that the perturbation $a \in {\rm W}^{1,\infty}(\re^2)$ and $q \in {\rm L}^{\infty}(\re^2)$ have their support in the set $\{ (x,y) ~|~
 |y| < R \}$, for some $0 < R < \infty$.
In addition, suppose that
the perturbation $a_2$ satisfies $\| a_2 \|_\infty \leq (c_n / 4)b^{1/2}$.
Then for any $\varphi \in {\rm Ran} ~ \mathbb{P}_{a,q}(\Delta_E(\delta))$, we have
\beq\label{eq:asympt-velocity1}
 \langle \varphi, V_y^\pm (\Delta) \varphi \rangle \geq \frac{c_n}{4} b^{1 \slash 2} \| \varphi \|^2 ,
\eeq
where the constant $c_n$ is defined in Proposition \ref{prop-me1}.
\end{thm}

The proof of Theorem \ref{th:asymptotic1} closely follows the proof in \cite[section 7]{his-soc1} (see also
\cite[section 4]{HS1}). We mention the main points.
We first prove the existence of the local wave operators \eqref{eq:local-wo1} for the pair
$H_0$ and  $H_1 = H (a,q)$, and the interval $\Delta_E (\delta)$, as in the theorem.
The key point is that in the application of the method of stationary phase, we use the positivity bound \eqref{a12}.
We then use the intertwining properties of the local wave operators to find
\beas
\langle \varphi, V_y^\pm (\Delta_E (\delta)) \varphi \rangle & = & \langle \varphi, \Omega_\pm (\Delta_E(\delta)) v_{y,a,q} {\Omega_\pm}
(\Delta_E(\delta))^* \varphi \rangle \nonumber \\
&=& \langle \Omega_\pm (\Delta_E(\delta))^* \mathbb{P}_1 (\Delta_E(\delta)) \varphi,  v_{y,a,q} {\Omega_\pm} (\Delta_E(\delta))^* \mathbb{P}_1(\Delta_E(\delta))\varphi \rangle \nonumber \\
 & = & \langle  \mathbb{P}_0 (\Delta_E(\delta)) \Omega_\pm (\Delta_E(\delta))^* \varphi,  v_{y,a,q} \mathbb{P}_0 (\Delta_E(\delta)) \Omega_\pm
 (\Delta_E(\delta))^* \varphi \rangle
\nonumber \\
& \geq & \frac{c_n}{4} b^{1 \slash 2}
 \| \mathbb{P}_0 (\Delta_E(\delta)) \Omega_\pm (\Delta_E(\delta))^* \varphi \|^2  ,
\eeas
where we used the lower bound \eqref{eq:edge-current-lb1} of Corollary \ref{corollary-edge-currents1}, the form of the current operator $v_{y, a, q}$ in \eqref{ec1}, and the estimate on $a_2$ given in the theorem.
To complete the proof, we again use the intertwining relation to write
$$
\| \P_0 (\Delta_E(\delta)) \Omega_\pm (\Delta_E(\delta))^* \varphi \| = \| \Omega_\pm (\Delta_E(\delta))^* \P_1 (\Delta_E(\delta)) \varphi \| = \| \varphi \| ,
$$
since the local wave operators are partial isometries.


\section{Asymptotic behavior of the eigenvalue counting function for negative perturbations of $H_0$ below $\inf \sigma_{\mathrm{ess}}(H_0)$}\label{sec:ev-counting1}
\setcounter{equation}{0}

In this section we apply the method introduced in \cite{R} to describe the discrete spectrum of the perturbed operator $H := H_0 - V$ near the infimum of its essential spectrum, when the scalar potential $V=V(x,y) \geq 0$ decays suitably as $|y| \to \infty$. For potentials of this type, we prove that there are an infinite number of eigenvalues accumulating at $\mathcal{E}_1 = \inf \sigma_{\mathrm{ess}}(H_0) = \inf \sigma_{\mathrm{ess}}(H) $ from below and we describe the behavior of the eigenvalue counting function.
The only information on $H_0$ we use here is the local behavior of the first band function $\omega_1(k)$ at its unique minimum $k=\kappa_1$. Namely, we recall from Proposition \ref{prop:effective-mass1} and the analyticity of $k \mapsto \omega_1(k)$ that the asymptotic identity
$$
\omega_1(k) - \EE_1 = \beta_1 (k-\kappa_1)^2 + O( (k-\kappa_1)^3 ),\ k \to \kappa_1,
$$
holds with $\beta_1:= \omega_1''(\kappa_1) \slash 2>0$.

\subsection{Statement of the result}\label{subsec:main1}

We first introduce the following notation. Let $H$ be a linear self-adjoint operator acting in a given separable Hilbert space. Assume that
$\EE = \inf \sigma_{\mathrm{ess}}(H) > -\infty$. The eigenvalue counting function $N(\mu ; H)$, $\mu \in (-\infty,\EE)$,
denotes the number of the eigenvalues of $H$ lying on the interval $(-\infty,\mu)$, and counted with the multiplicities.
We recall that $\psi_1(x,k)$ is the first eigenfunction of the fiber operator $h(k)$ with band function $\omega_1 (k)$.

\begin{thm}
\label{thm-ea}
Let $V (x,y) \in {\rm L}^{\infty}(\re^2)$ satisfy the following two conditions:
\begin{enumerate}
\item[i.)] $\exists (\alpha,C) \in (0,2) \times \R_+^*$ so that
$$
 0 \leq V(x,y) \leq C(1 + |x|)^{-\alpha}(1+|y|)^{-\alpha},\ x, y  \in \re;
$$
\item[ii.)] $\exists L>0$ so that $\lim_{|y| \to \infty} |y|^{\alpha} \int_{\re} V(x,y) \psi_1(x,\kappa_1)^2 dx =L$.
\end{enumerate}
Then we have
\bel{ea2}
\lim_{\lambda \downarrow 0} \lambda^{-\frac{1}{2} + \frac{1}{\alpha}} N(\EE_1-\lambda ; H_0 -
V) =\frac{2}{\alpha \pi}
\beta_1^{-1/2} L^{1/\alpha}
{\rm B}\left(\frac{3}{2},\frac{1}{\alpha}-\frac{1}{2}\right),
\eeq
where
${\rm B}(\cdot, \cdot)$ is the Euler beta function \cite[section 6.2]{AS}
and $\beta_1:= \omega_1''(\kappa_1) \slash 2>0$ is the effective  mass.
\end{thm}


\subsection{Some notation and auxiliary results}\label{subsec:aux1}

This subsection presents some notation and several auxiliary results needed
for the proof of Theorem \ref{thm-ea}, which is presented in \S \ref{sec-proofthmea}.

For a linear compact self-adjoint operator $H$
acting in a separable Hilbert space, we define
$$
n(s;H): = {\rm rank}\; {\mathbb P}_{(s,\infty)}(H),\ s>0,
$$
where ${\mathbb P}_I(H)$ denotes the
spectral projection of $H$ associated with  the
interval $I \subset \re$.
Let $X_1$ and $X_2$ be two separable Hilbert spaces. For a linear compact operator $H : X_1 \to X_2$,
we set
\bel{ea3b}
\mathfrak{n}(s;H): = n(s^2;H^*H),\ s>0.
\ee
If $H_j : X_1 \to X_2$, $j=1,2$, are two linear compact operators, we will use Ky Fan inequality
\bel{ea4}
 \mathfrak{n}(s_1+s_2,H_1+H_2) \leq  \mathfrak{n}(s_1,H_1)+  \mathfrak{n}(s_2,H_2) ,
\ee
which holds for $s_1>0$ and $s_2>0$ according to \cite[Chapter I, Eq. (1.31)]{BSo} and
\cite[Chapter II, Section 2, Corollary 2.2]{GK}.

For further reference, we recall from \cite[Eq. (2.1) \& Lemma 2.3]{R} the following technical result.
\begin{lemma}\cite[Lemma 2.3]{R}
\label{lm-ea1}
Let $G : {\rm L}^2(\re) \to {\rm L}^2(\re^2)$ be a bounded  operator with integral kernel $g \in
  {\rm L}^{\infty}(\re^3)$. Then for every $f \in {\rm L}^r(\re^2)$ and $h
  \in {\rm L}^r(\re)$ with $r \in [2,\infty)$, we have
$$
\mathfrak{n}(s; fGh) \leq C_r(G) s^{-r} \|f\|_{{\rm L}^r(\re^2)}^r \|h\|_{{\rm L}^r(\re)}^r,\ s>0,
$$
where $C_r(G):= \|g\|^{4 \slash r}_{{\rm L}^{\infty}(\re^3)} \|G\|^{2(r-2) \slash r}$.
\end{lemma}
For $\delta>0$ fixed, let $\chi = \chi_{\delta}$ denote the
characteristic function of the interval $I=I_{\delta}:=(\kappa_1-\delta,\kappa_1+\delta)$.
As we shall actually apply Lemma \ref{lm-ea1} in \S \ref{sec-proofthmea} with $G=\Gamma_j$, $j=0,1$, where $\Gamma_j : {\rm L}^2(\re) \to
{\rm L}^2(\re^2)$ is the integral operator with kernel
\bel{ea6a}
\gamma_0(x,y;k) : = \frac{1}{\sqrt{2\pi}} \psi_1(x,\kappa_1) e^{-iyk} \chi(k) ,\ (x,y) \in \re^2,\ k \in \re,
\ee
and
\bel{ea6b}
\gamma_1(x,y;k) : = \frac{1}{\sqrt{2\pi}} \left( \frac{\psi_1(x,k) - \psi_1(x,\kappa_1)}{k-\kappa_1} \right) e^{-iyk}
\chi(k),\ (x,y) \in \re^2,\ k \in \re\setminus \{ \kappa_1\}.
\ee

\begin{lemma}
\label{lm-ea3}
We have $\gamma_j \in {\rm L}^{\infty}(\re^3)$ for $j=0,1$.
\end{lemma}
\begin{proof}
In view of \eqref{ea6a}-\eqref{ea6b},
it suffices to prove that $(x,k) \mapsto \psi_1(x,k)$ and $(x,k) \mapsto (\psi_1(x,k)-\psi_1(x,\kappa_1)) \slash (k-\kappa_1)$ are respectively bounded in $\R \times I$ and $\R \times ( I \setminus \{ \kappa_1 \} )$.
The eigenfunction $\psi_1(\cdot , k)$ is a solution to the second order ordinary differential equation
\bel{ea6d}
-\varphi''(x) + W(x;k) \varphi(x)=0,\ x \in \re,
\ee
where $W(x;k):=(b |x| -k)^2 - \omega_1(k)$. The potential $W(x;k)$
is greater than $\delta^2$ provided $|x| > x_1 :=(b^{1 \slash 2} + \kappa_1 + 2 \delta ) \slash b$, uniformly in $k \in I$.
It follows from \cite[Lemma B.3]{HS1} that
$$
0 < \psi_1(x,k) \leq \psi_1( \pm x_1,k) e^{-\delta (|x|-x_1)},\ x \geq x_1,\ k \in I.
$$
Since $(x,k) \mapsto \psi_1(x,k)$ is continuous in $\R \times I$, this implies the result for $j=0$.

Next, bearing in mind that the ${\rm L}^2(\R)$-valued function $k \mapsto \psi_1(\cdot,k)$ is real analytic, we deduce from \eqref{ea6d} that $\Phi(\cdot,k) := \partial_k \psi_1(\cdot,k)$ is solution to the equation
$$ -\varphi''(x) + W(x;k) \varphi(x) = -F(x;k),\ x \in \R, $$
where $F(x;k) := \partial_k W(x;k) \psi_1(x,k) = (2(k - b|x|) - \omega_1'(k))  \psi_1(x,k)$. Therefore we get that
\bea
& & \| \Phi'(\cdot,k) \|_{{\rm L}^2(\R)}^2 + \| (b|x|-k) \Phi(\cdot,k) \|_{{\rm L}^2(\R)}^2 \nonumber \\
& \leq & ( C+ \omega_1(k) \| \Phi(\cdot,k) \|_{{\rm L}^2(\R)}) \| \Phi(\cdot,k) \|_{{\rm L}^2(\R)}, \label{ea6c}
\eea
with $C:=\sup_{k \in I} \| F(\cdot;k) \|_{{\rm L}^2(\R)}<\infty$, by standard computations. Since $\sup_{k \in I} \| \Phi(\cdot,k) \|_{{\rm L}^2(\R)} < \infty$, \eqref{ea6c} thus entails that $\sup_{k \in I} \| \Phi(\cdot,k) \|_{{\rm H}^1(\R)} < \infty$. From this and the estimate
$$ \Phi(x,k)^2 = 2 \int_{-\infty}^x \Phi(x,k) \Phi'(x,k) dx \leq \| \Phi(\cdot,k) \|_{{\rm H}^1(\re)}^2,\ x \in \R,\ k \in \R, $$
then follows that $\sup_{(x,k) \in \re \times I} | \Phi(x,k) | < \infty$. This yields the result for $j=1$ and terminates the proof.
\end{proof}

Finally, since the proof of Theorem \ref{thm-ea} is obtained by expressing $\lim_{\lambda \downarrow 0} N(\EE_1 - \lambda;H_0-V)$ in terms of the asymptotics of the eigenvalue counting function for the
discrete spectrum of a second-order ordinary differential operators on the real line, we recall from \cite[Lemma 4.9]{BKRS} the
following
\begin{lemma}
\label{lm-ea2}
Assume that $\mathcal{Q} = \overline{\mathcal{Q}} \in {\rm L}^{\infty}(\re)$ satisfies the two following conditions:
\begin{enumerate}
\item[i.)] $\exists (\alpha,C) \in (0,2) \times \R_+^*$ so that $|\mathcal{Q}(x)| \leq C (1 + |x|)^{-\alpha},\ x \in \re$;
\item[ii.)] $\exists \ell>0$ so that $\lim_{|x| \to \infty} |x|^{\alpha} \mathcal{Q}(x) =\ell$.
\end{enumerate}
For any $\mathfrak{m} > 0$, let ${\mathcal H}(\mathfrak{m}, \mathcal{Q}): = - \mathfrak{m}^2 \frac{d^2}{dx^2} - \mathcal{Q}$
be the 1D Schr\"odinger operator with domain ${\rm H}^2(\re)$, self-adjoint in ${\rm L}^2(\re)$. \\
Then we have
$$\lim_{\lambda\downarrow 0}\,  \lambda^{\frac 1\alpha-\frac 12}\,
N(-\lambda;{\mathcal H}(\mathfrak{m}, \mathcal{Q})) = \frac{2 \ell^{\frac
1\alpha}}{\pi\alpha \mathfrak{m}} B\left(\frac 32,\, \frac 1\alpha-\frac
12\right).$$
\end{lemma}
The proof of Lemma \ref{lm-ea2}, which is similar to the one of \cite[Theorem XIII.82]{RS-IV}, can be found in \cite{LS}.


\subsection{Proof of Theorem \ref{thm-ea}}
\label{sec-proofthmea}
The proof consists of four parts.

\subsubsection{Part I: Projection on the bottom of the first band function}
\label{ssubsec:part1}
We define ${\mathcal V} : = \FF V {\mathcal{F}}^*$ and recall that $\mathcal{H}_0 = \FF H_0 {\mathcal{F}}^*$.
The first part of the proof is to show that the asymptotics of $N( \EE_1-\lambda ; H_0 - V)$ as $\lambda \downarrow 0$ is determined by the asymptotics of the eigenvalue counting function for a reduced operator obtained from the projection of the operator ${\mathcal H}_0 -{\mathcal V}$ to the bottom of the first band function.
First of all, we remark
that the multiplier by $V$ is $H_0$-compact since $V(x,y)$ goes to zero as $|(x,y)|$ tends to infinity. As a consequence we have
$$
\inf \sigma_{\rm{ess}}(H_0 - V) = \inf \sigma_{\rm ess}(H_0),
$$
hence $N(\EE_1-\lambda ; H_0-V ) < \infty$ for any $\lambda>0$. Furthermore, since $V \in {\rm L}^\infty(\R^2)$, the operator $H_0 - V$ is lower semibounded. The operator $H_0 - V$ is unitarily equivalent to ${\mathcal H}_0 -
{\mathcal V}$, so we have
$$
N( \EE_1-\lambda ; H_0 - V ) = N( \EE_1-\lambda ; {\mathcal H}_0 -
{\mathcal V}),\ \lambda > 0.
$$
Let $P:{\rm L}^2(\re^2) \to
{\rm L}^2(\re^2)$ be the orthogonal projection defined by
\bel{ea7b}
(Pu)(x;k) : =  \left( \int_{\re} u(t;k) \psi_1(t,k) dt \right) \chi(k) \psi_1(x,k),\
(x,k) \in \re^2,
\ee
where we recall that $\chi$ denotes the
characteristic function of the interval $I:=(\kappa_1-\delta,\kappa_1+\delta)$ for some fixed $\delta>0$.
\begin{lemma}
\label{lm-ea4}
Let ${\mathcal H}_1(t)$, $t \in \re$, be the operator $P({\mathcal H}_0 -
(1+t){\mathcal V})P$ with domain $P \Dom({\mathcal H}_0)$.
Then there is a constant $N_0 \geq 0$, independent of $\lambda$, such that we have
\bel{ea7c}
N(\EE_1-\lambda;\HH_1(0)) \leq N(\EE_1-\lambda;\HH_0-\VV) \leq N(\EE_1-\lambda;\HH_1(0)) + N_0,\ \lambda >0.
\ee
\end{lemma}
\begin{proof}
Set $Q : = I - P$. For all $u \in {\rm L}^2(\re^2)$ and $\eps>0$ it holds true that
\beas
 & & | \langle (P \VV Q +Q \VV P) u , u \rangle |  = 2 \left| \Pre{\langle \VV^{1 \slash 2} P u ,  \VV^{1 \slash 2} Q u \rangle} \right| \\
& \leq & 2 \| \VV^{1 \slash 2} P u \| \| \VV^{1 \slash 2} Q u \|
\leq \eps \langle P \VV P u , u \rangle + \eps^{-1} \langle Q \VV Q u , u \rangle,
\eeas
which entails
$$ -\eps P \VV P - \eps^{-1} Q \VV Q \leq P \VV Q + Q \VV P \leq \eps P \VV P + \eps^{-1} Q \VV Q,$$
in the sense of quadratic forms. From this and the elementary identity $\mathcal{H}_0 = P \mathcal{H}_0 P + Q \mathcal{H}_0 Q$ then follows that
\bel{ea8b}
{\mathcal H}_1(\eps) \oplus {\mathcal H}_2(\eps) \leq \HH_0 - \VV \leq {\mathcal H}_1(-\eps) \oplus {\mathcal H}_2(-\eps),\ \eps>0,
\ee
where ${\mathcal H}_2(t)$, $t \in \re^*$, is the operator $Q({\mathcal H}_0 - (1+t^{-1}){\mathcal V})Q$ with domain $Q \Dom({\mathcal H}_0)$, and the symbol $\oplus$ indicates an orthogonal sum.
Therefore, for every $\lambda >0$ and $\eps>0$ fixed, the left inequality in \eqref{ea8b} implies
\bel{ea8c}
N(\EE_1-\lambda ; \HH_0-\VV ) \leq N(\EE_1-\lambda ; {\mathcal H}_1(\eps)) + N(\EE_1-\lambda ; {\mathcal H}_2(\eps)),
\ee
while the right one yields
\bel{ea8d}
N(\EE_1-\lambda ; \HH_0-\VV) \geq N(\EE_1-\lambda ; {\mathcal H}_1(-\eps)) + N(\EE_1-\lambda ; {\mathcal H}_2(-\eps)) \geq N(\EE_1-\lambda ; {\mathcal H}_1(-\eps)).
\ee
Further, the multiplier by $\VV$ being $\mathcal{H}_0$-compact, $Q \VV Q$ is $Q \HH_0$-compact and
\bel{ea8e}
\sigma_{\rm ess}(\HH_2(\eps))=\sigma_{\rm ess}(Q \HH_0),\ \eps>0.
\ee
On the other hand we have $\inf \sigma_{\rm ess}(Q \HH_0) = \min\{\omega_1(\kappa_1+\delta), \min_{k \in \re} \omega_2(k)\} > \EE_1$ hence
$$
N_0:=N(\EE_1;Q \HH_0) < \infty,
$$
and \eqref{ea8e} yields
\bel{ea8f}
N(\EE_1-\lambda;\HH_2(\eps)) = N(\EE_1-\lambda; Q\HH_0) \leq N_0,\ \lambda >0,\ \eps>0.
\ee
Putting \eqref{ea8c} and \eqref{ea8f} together, we get that
\bel{ea9}
N(\EE_1-\lambda ; \HH_0-\VV ) \leq N(\EE_1-\lambda ; {\mathcal H}_1(\eps)) + N_0,\ \lambda>0,\ \eps>0.
\ee
Letting $\eps \downarrow 0$ in \eqref{ea8d} and \eqref{ea9}, we obtain \eqref{ea7c}.
\end{proof}

\subsubsection{Part II: Singular integral operator decomposition}
\label{ssubsec:part2}
This part involves relating the number of eigenvalues accumulating below the bottom of the essential spectrum of ${\mathcal H}_1(0)$, to the local behavior of $\omega_1(k)$ and $\psi_1(\cdot,k)$ at $\kappa_1$.

The main tool we use for this is the Birman-Schwinger principle, which, in this situation, implies
\bel{ea10a}
N( \EE_1-\lambda ; {\mathcal H}_1(0)) = n(1; P
({\mathcal H}_0 - \EE_1 + \lambda)^{-1/2} {\mathcal V} ({\mathcal H}_0 - \EE_1 + \lambda)^{-1/2}P).
\ee
In view of \eqref{ea7b} and \eqref{ea10a}, we set
\beq\label{eq:quad0}
a(k;\lambda) : = (\omega_1(k) - \EE_1 + \lambda)^{-1/2},\ k \in \re,\
\lambda > 0,
\eeq
and denote by $\Gamma : {\rm L}^2(\re_k) \to {\rm L}^2(\re_{x,y}^2)$ the operator with integral kernel
$$
\gamma(x,y;k) : = \frac{1}{\sqrt{2\pi}} \psi_1(x,k) \mathrm{e}^{iy k} \chi(k),\ (x,y) \in \re^2,\ k \in \re.
$$
For every $\lambda>0$ the operator $\chi a(\lambda)\Gamma^* V \Gamma a(\lambda) \chi$ is
self-adjoint and nonnegative in ${\rm L}^2(\re_k)$. Furthermore we get
\bel{ea10b}
P({\mathcal H}_0 - \EE_1 + \lambda)^{-1/2} {\mathcal V} ({\mathcal H}_0 - \EE_1 + \lambda)^{-1/2}P
= \UU^* \chi a(\lambda)\Gamma^* V \Gamma a(\lambda) \chi \UU,
\ee
by direct calculation, where $\UU : \Ran P \to {\rm L}^2(I)$ is the unitary transform
$$ (\UU f)(k) := \left( \int_{\re} f(x,k) \psi_1(x,k) dx \right) \chi(k),\ k \in \re. $$
From \eqref{ea10a}-\eqref{ea10b} then follows that
\bel{ea10}
N( \EE_1-\lambda ;  {\mathcal H}_1(0)) =  n(1; \chi a(\lambda)\Gamma^* V \Gamma a(\lambda) \chi),\ \lambda>0.
\ee
Putting $W : = V^{1/2}$ we deduce from \eqref{ea3b} and \eqref{ea10} that
\bel{ea11}
N( \EE_1-\lambda ;  {\mathcal H}_1(0))=\mathfrak{n}(1; W \Gamma a(\lambda)
\chi),\ \lambda>0.
\ee

\subsubsection{Part III: Reduction to the quadratic leading term of the first band function}
Due to \eqref{ea7c} and \eqref{ea11}, we are left with the task of computing the asymptotics of $\mathfrak{n}(1; W \Gamma a(\lambda)
\chi)$ as $\lambda \downarrow 0$. In this subsection, we shall prove that $\Gamma$ and $a(\lambda)\chi$ may be replaced by, respectively,
$\Gamma_0$ and $\mathfrak{a}(\lambda) \chi$, in the above expression. The operator
$\Gamma_0: {\rm L}^2(\re) \to {\rm L}^2(\re^2)$ is the operator with integral kernel
$\gamma_0(x,y;k) $ given by \eqref{ea6a}. We obtain $\mathfrak{a}(\lambda)$ from $a(\lambda)$ in \eqref{eq:quad0}
by replacing $\omega_1(k)$ by the first two terms of the expansion of $\omega_1(k)$ about $\kappa_1$:
\beq\label{eq:quad1}
\mathfrak{a}(k;\lambda) : = \left(\beta_1 (k-\kappa_1)^2+ \lambda \right)^{-1/2},\ k \in \re,\ \lambda > 0.
\eeq

\begin{lemma}
\label{lm-ea5}
Let $r \geq 2$ fulfill $r > 2 \slash \alpha$. Then there exists a constant $N_r \geq 0$ such
that the estimates
\bea
\mathfrak{n}((1+\eps)^3;W \Gamma_0 \mathfrak{a}(\lambda)) - N_r \eps^{-r}
& \leq & \mathfrak{n}(1; W \Gamma a(\lambda) \chi) \nonumber \\
& \leq & \mathfrak{n}((1-\eps)^2;W \Gamma_0 \mathfrak{a}(\lambda)) + N_r \eps^{-r}, \label{ea11b}
\eea
hold for all $\lambda >0$ and $\eps \in (0,1)$. 
\end{lemma}
\begin{proof}
1. We use the decomposition $\Gamma a(\lambda) \chi = \sum_{j=0}^1 \Gamma_j a_j(\lambda) \chi$, where
$\Gamma_1 : {\rm L}^2(\re) \to {\rm L}^2(\re^2)$ is the operator with integral kernel
$\gamma_1(x,y;k)$ defined in \eqref{ea6b}, and
$$
a_j(k;\lambda) := (k-\kappa_1)^j a(k;\lambda),\ j=0,1.
$$
Since $\gamma_j \in {\rm L}^{\infty}(\re^3)$, $j=0,1$, by Lemma \ref{lm-ea3}, the operators $\Gamma_j$ are bounded with
$$ \|\Gamma_0\| =1\ {\rm and}\
\|\Gamma_1\|^2 \leq \sup_{k \in I} \int_{\re} \left( \frac{\psi_1(x,k) - \psi_1(x,\kappa_1)}{k-\kappa_1} \right)^2 dx.
$$
We notice from \eqref{ea4} that
\bea
 & & \mathfrak{n}(1+\eps; W \Gamma_0 a_0(\lambda) \chi) - \mathfrak{n}(\eps; W \Gamma_1 a_1(\lambda) \chi) \nonumber \\
& \leq & \mathfrak{n}(1; W \Gamma a(\lambda) \chi) \nonumber \\
& \leq & \mathfrak{n}(1-\eps; W \Gamma_0 a_0(\lambda) \chi) + \mathfrak{n}(\eps; W \Gamma_1 a_1(\lambda) \chi),\ \lambda>0,\ \eps \in (0,1). \label{ea12}
\eea

\noindent
2. We obtain an upper bound for $\mathfrak{n}(\eps; W \Gamma_1 a_1(\lambda) \chi)$ in \eqref{ea12} from Lemma \ref{lm-ea1} taking $G=\Gamma_1$, $f=W$, and $h=a_1(\lambda)\chi$.
We get that
\beas
\mathfrak{n}(\eps; W \Gamma_1 a_1(\lambda) \chi) & \leq &
C_r(\Gamma_1) \eps^{-r} \|W\|_{{\rm L}^r(\re^2)}^r \|a_1(\lambda) \chi\|_{{\rm L}^r(\re)}^r \nonumber \\
& \leq & n_r \eps^{-r},\ \lambda>0,\ \eps \in (0,1),  
\eeas
with $n_r:= C_r(\Gamma_1) \|W\|_{{\rm L}^r(\re^2)}^r \|a_1(0) \chi\|_{{\rm L}^r(\re)}^r$.
From this and \eqref{ea12} then follows that
\bea
 \mathfrak{n}(1+\eps; W \Gamma_0 a_0(\lambda) \chi) - n_r \eps^{-r}
& \leq & \mathfrak{n}(1; W \Gamma a(\lambda) \chi)  \nonumber \\
& \leq &  \mathfrak{n}(1-\eps; W \Gamma_0 a_0(\lambda) \chi) + n_r \eps^{-r},  \label{ea13}
\eea
for $\lambda >0,\ \eps \in (0,1)$.

\noindent
3. Next, recalling that $\chi$ is the characteristic function of the interval
$( \kappa_1 - \delta, \kappa_1 + \delta)$, for $\eps \in (0,1)$ fixed, we choose $\delta > 0$ so small that
$$
(1+\eps)^{-1} \mathfrak{a}(k;\lambda) \chi(k) \leq a_0(k;\lambda)
\chi(k) \leq
(1-\eps)^{-1} \mathfrak{a}(k;\lambda)\chi(k),\ k \in \re,\ \lambda > 0 ,
$$
where $\mathfrak{a}(\lambda)$ is defined in \eqref{eq:quad1}. It follows from this and the simple identity that $n(s,tH) = n(t^{-1}s,H)$,
for $s,t>0$, that we have
\bea\label{ea14}
\mathfrak{n}(s(1+\eps); W \Gamma_0 \mathfrak{a}(\lambda) \chi) & \leq & \mathfrak{n}(s ; W
\Gamma_0 a_0(\lambda) \chi)  \nonumber \\
& \leq  & \mathfrak{n}(s(1-\eps); W \Gamma_0
\mathfrak{a}(\lambda) \chi),\ s>0,
\eea
Moreover \eqref{ea4} and the minimax principle yield
\bea
\mathfrak{n}(s(1+\eps); W \Gamma_0 \mathfrak{a}(\lambda)) - \mathfrak{n}(s
\eps; W \Gamma_0 \mathfrak{a}(\lambda) (1-\chi))
& \leq & \mathfrak{n}(s ; W \Gamma_0 \mathfrak{a}(\lambda) \chi) \nonumber \\
&\leq & \mathfrak{n}(s; W \Gamma_0 \mathfrak{a}(\lambda)), \nonumber \\
 & &  \label{ea15}
\eea
for $s>0, \eps \in (0,1)$, as we have $\mathfrak{a}(\lambda)  \Gamma_0^* V \Gamma_0 \mathfrak{a}(\lambda) \geq \chi \mathfrak{a}(\lambda)  \Gamma_0^* V \Gamma_0 \mathfrak{a}(\lambda) \chi$ in the sense of quadratic forms.
Combining the second inequality of \eqref{ea14} with $s = 1 - \eps$ with the second inequality of
\eqref{ea15} with $s = (1-\eps)^2$, we obtain
\bea
 \mathfrak{n}(1-\eps; W \Gamma_0 a_0(\lambda) \chi)
& \leq & \mathfrak{n}((1-\eps)^2;W \Gamma_0 \mathfrak{a}(\lambda) \chi) \nonumber \\
& \leq & \mathfrak{n}((1-\eps)^2;W \Gamma_0 \mathfrak{a}(\lambda)),\ \lambda >0. \label{ea15b}
\eea
Similarly, combining the first inequality of \eqref{ea15} with $s = (1 + \eps)^2$ with the first inequality of
\eqref{ea14} for $s=1+\eps$, we find that
\bea
& & \mathfrak{n}((1+\eps)^3;W \Gamma_0 \mathfrak{a}(\lambda)) - \mathfrak{n}(\eps (1+ \eps)^2;W \Gamma_0 \mathfrak{a}(\lambda) (1 - \chi)) \nonumber \\
& \leq & \mathfrak{n}((1+\eps)^2;W \Gamma_0 \mathfrak{a}(\lambda) \chi) \nonumber \\
& \leq & \mathfrak{n}(1+\eps; W \Gamma_0 a_0(\lambda) \chi),\ \lambda>0. \label{ea15c}
\eea

\noindent
4. In order to evaluate $\mathfrak{n}(\eps (1+\eps)^2; W \Gamma_0 \mathfrak{a}(\lambda)(1-\chi))$ in \eqref{ea15c},
we use Lemma \ref{lm-ea1} with $G=\Gamma_0$, $f=W$ and
$h=\mathfrak{a}(\lambda)(1-\chi)$. We obtain that
\bea
& & \mathfrak{n}(\eps (1+\eps)^2; W \Gamma_0 \mathfrak{a}(\lambda) (1-\chi)) \nonumber \\
& \leq &
C_r(\Gamma_0) \eps^{-r} (1+\eps)^{-2r}\|W\|_{{\rm L}^r(\re^2)}^r \|\mathfrak{a}(\lambda)
(1-\chi)\|_{{\rm L}^r(\re)}^r \nonumber \\
& \leq & n_r' \eps^{-r},\  \lambda>0, \label{ea15d}
\eea
with $n_r':=C_r(\Gamma_0) \|W\|_{{\rm L}^r(\re^2)}^r \|\mathfrak{a}(0)
(1-\chi)\|_{{\rm L}^r(\re)}^r$.
Finally, \eqref{ea11b} follows from \eqref{ea13} and \eqref{ea15b}--\eqref{ea15d} upon setting $N_r:=n_r+n_r'$.
\end{proof}

Summing up \eqref{ea7c} and \eqref{ea11}-\eqref{ea11b}, we have so far derived the following upper bound:
\bea
& & \mathfrak{n}((1+\eps)^3;W \Gamma_0 \mathfrak{a}(\lambda) ) - N_r \eps^{-r} \nonumber \\
& \leq & N(\EE_1 - \lambda ; H_0 - V ) \nonumber \\
& \leq & \mathfrak{n}((1-\eps)^2;W \Gamma_0 \mathfrak{a}(\lambda) ) +N_0 + N_r \eps^{-r},\ \lambda >0,\ \eps \in (0,1). \label{ea16}
\eea

\subsubsection{Part IV: Reduction to a 1D problem}
Let $\mathfrak{h}(s)$, $s>0$, be the Hamiltonian $\mathcal{H}(\mathfrak{m},\mathcal{Q})$ introduced in Lemma \ref{lm-ea2}, with
$\mathfrak{m}:=\beta_1^{1 \slash 2}$ and $\mathcal{Q}:=s^{-2} \int_{\re} V(x,y) \psi_1(x,\kappa_1)^2 dx$.
By the Birman-Schwinger principle, we have
$$
\mathfrak{n}(s; W \Gamma_0 \mathfrak{a}(\lambda)) = n(1; s^{-2} \mathfrak{a}(\lambda)
\Gamma_0^* V \Gamma_0  \mathfrak{a}(\lambda))
 = N( -\lambda ; \mathfrak{h}(s)),\ s>0,\ \lambda >0.
$$
Lemma \ref{lm-ea2} applied to the Hamiltonian $\mathfrak{h}(s)$, $s>0$ yields the asymptotic
\beq\label{eq:asympt1}
\lim_{\lambda \downarrow 0} \lambda^{-\frac{1}{2} + \frac{1}{\alpha}}
N( -\lambda ; \mathfrak{h}(s)) = c(\alpha,\beta_1,L) s^{-2/\alpha},\ s>0,
\eeq
with $c(\alpha,\beta_1,L):= 2 \slash (\alpha \pi)
\beta_1^{-1/2} L^{1 \slash \alpha}
{\rm B}\left(3 \slash 2,1 \slash \alpha-1 \slash 2 \right)$.
To obtain a lower bound on $N(\EE_1 - \lambda ; H_0 - V ) $ from the first inequality in \eqref{ea16}, we take $s = (1 + \eps)^3$
in \eqref{eq:asympt1} and obtain
\bel{ea19}
\liminf_{\lambda \downarrow 0} \lambda^{-\frac{1}{2} +
\frac{1}{\alpha}} N(\EE_1-\lambda ; H_0-V) \geq c(\alpha,\beta_1,L) (1+\varepsilon)^{-6/\alpha},\ \eps \in (0,1).
\ee
The upper bound is obtained in a similar manner taking $s = (1-\eps)^2$ in \eqref{eq:asympt1},
\bel{ea20}
\limsup_{\lambda \downarrow 0} \lambda^{-\frac{1}{2} +
\frac{1}{\alpha}} N(\EE_1-\lambda ; H_0-V) \leq
c(\alpha,\beta_1,L) (1-\varepsilon)^{-4/\alpha},\ \eps \in (0,1).
\ee
Letting $\eps \downarrow 0$ in \eqref{ea19}-\eqref{ea20}, we obtain \eqref{ea2}. This completes the proof of Theorem \ref{thm-ea}.


\end{document}